\documentclass[journal]{IEEEtran}
\usepackage{csquotes}
\usepackage{amsmath}
\usepackage{amsmath, amsthm, amssymb}
\usepackage{csquotes}
\usepackage{amsmath}

\usepackage{authblk}
\usepackage{hyperref}
\usepackage{float}
\usepackage{cite}
\usepackage{algorithm}
\usepackage{algorithmicx, algpseudocode}
\usepackage{multirow}

\usepackage{mathtools}
\ifCLASSINFOpdf
\usepackage[pdftex]{graphicx}
\else
\fi
\usepackage{amsmath}
\usepackage{amsfonts}
\usepackage{amssymb}
\usepackage{color}
\usepackage{array}
\usepackage{fixltx2e}
\usepackage{stfloats}
\usepackage{url}
\usepackage{amsmath, amsthm, amssymb}
\newtheorem{theorem}{Theorem}

\newtheorem{corollary}{Corollary}

\usepackage{booktabs}
\usepackage{varwidth}
\usepackage[table,xcdraw]{xcolor}

\begin{document}

\title{\linespread{1} Multiple RISs-Aided Networks: Performance Analysis and Optimization}
\author{Mahmoud~Aldababsa,~\IEEEmembership{}
	Anas~M.~Salhab,~\IEEEmembership{Senior~Member,~IEEE,}
	Ali~Arshad~Nasir,~\IEEEmembership{Senior~Member,~IEEE,}
	Monjed H. Samuh,~\IEEEmembership{}
	and~Daniel Benevides~da~Costa,~\IEEEmembership{Senior~Member,~IEEE}
	\thanks{M. Aldababsa is with the Department of Electrical and Electronics Engineering, Istanbul Gelisim University, 34310, Istanbul, Turkey (e-mail: mhkaldababsa@gelisim.edu.tr).}
	\thanks{A. M. Salhab and A. A. Nasir are with the Department of Electrical Engineering, King Fahd University of Petroleum \& Minerals, Dhahran 31261, Saudi Arabia (e-mail: salhab@kfupm.edu.sa, anasir@kfupm.edu.sa).}
	\thanks{M. H. Samuh is with the Department of Applied Mathematics \& Physics, Palestine Polytechnic University, Hebron, Palestine (e-mail: monjedsamuh@ppu.edu).}
	\thanks{D. B. da Costa is with the AI \& Telecom Research Center, Technology Innovation Institute, 9639 Masdar City, Abu Dhabi, United Arab Emirates (email: danielbcosta@ieee.org).}
	\thanks{}
	\thanks{}}

\maketitle
\begin{abstract}
This paper analyzes the performance of multiple reconfigurable intelligent surfaces (RISs)-aided networks. The paper also provides some optimization results on the number of reflecting elements on RISs and the optimal placement of RISs. We first derive accurate closed-form approximations for RIS channels' distributions assuming independent non-identically distributed (i.ni.d.) Nakagami-\emph{m} fading environment. Then, the approximate expressions for outage probability (OP) and average symbol error probability are derived in closed-form. Furthermore, to get more insights into the system performance, we derive the asymptotic OP at the high signal-to-noise ratio regime and provide closed-form expressions for the system diversity order and coding gain. Finally, the accuracy of our theoretical analysis is validated through Monte-Carlo simulations. The obtained results show that the considered RIS scenario can provide a diversity order of $\frac{a}{2}K$, where $a$ is a function of the Nakagami fading parameter $m$ and the number of meta-surface elements $N$, and $K$ is the number of RISs. 
\end{abstract}
\IEEEpeerreviewmaketitle
\section{Introduction}

\IEEEPARstart{R}{econfigurable} intelligent surfaces (RISs) have recently attracted noticeable attention as a promising candidate for future wireless communication networks \cite{sb2021}. An RIS is an artificial surface composed of massive low-cost reconfigurable passive elements. It can reconfigure the propagation of incident wireless signals by adjusting the amplitude and phase shift of each element \cite{Renzo2}. As RISs do not require radio frequency (RF) chains, this remarkably reduces energy consumption and hardware costs, making RISs more economical and environmentally friendly than multi-antenna and relaying systems \cite{Renzo1}-\hspace{-0.01cm}\cite{Alouini1}. Accordingly, RISs represent a new low-cost/less-complicated solution to realize wireless communication with high spectral and energy efficiencies. Driven by these advantages, RISs have received extensive interest from both industry and academy to exploit their benefits fully. In the literature, single-RIS-aided systems have been extensively studied \cite{Liang}-\hspace{-0.01cm}\cite{Ferreira}. In \cite{Liang}, an overview of the basic characteristics of the RIS/antenna technology and its potential applications has been provided. A detailed overview on the state-of-the-art solutions, fundamental differences of RIS with other technologies, and the most important open research issues in this research area has been provided in \cite{Basar1}. The authors in \cite{Wu} showed that RIS has better performance than conventional massive multiple-input multiple-output systems as well as multi-antenna amplify-and-forward (AF) relaying networks while reducing the system complexity and cost. 

Recently, Yang \textit{et al.} studied in \cite{Yang1} the performance of RIS-assisted mixed indoor visible light communication/radio frequency (RF) system. They derived closed-form expressions for the outage probability (OP) and bit error rate (BER) of AF and decode-and-forward (DF) relaying schemes. In \cite{Yang4}, the authors utilized RIS to improve the quality of the source signal, which is sent to the destination through an unmanned aerial vehicle. The secrecy OP of RIS-assisted network has been derived in \cite{Yang5} in the presence of a direct link and eavesdropper. It is worthwhile to mention here that some of the previous works on RIS-assisted networks performed their analysis based on the central limit theorem (CLT), which makes it applicable only for a large number of reflecting elements \cite{Yang6}. Following other approaches and to cover any number of reflecting elements, \cite{Yang6} has derived accurate approximations for the channel distributions and performance metrics of RIS-assisted networks assuming Rayleigh fading channels. Closed-form expressions for the OP, symbol error rate (SER), and ergodic capacity of the same system were presented in \cite{Boulogeorgos}. Most recently, \cite{Ferreira} provided closed-form expressions for the bit error probability of RIS-assisted network over Nakagami-$m$ fading channels. As stated by the authors, their results are valid only for BPSK and $M$-QAM modulation techniques. Although the authors in \cite{Ferreira} considered Nakagami-$m$ fading channels, they only derived exact expressions for the error probability for the limited number of reflecting elements. In addition, they derived their approximate expressions and bounds for two specific modulation schemes, as mentioned before. Furthermore, no insights into the system performance at high signal-to-noise ratio (SNR) values were provided in that study. 

Away from single RIS networks, a few works in literature considered the case of multiple RISs in their analysis. In this regard, the existing papers could be categorized into two types: optimization papers and performance analysis papers. One of the works that proposes and solves an optimization problem for RIS-aided networks is \cite{Fang}. In that paper, the authors proposed new optimum location-based RIS selection policies in RIS-aided wireless networks to maximize the SNR for a power path-loss model in outdoor communications and an exponential path-loss model in indoor communications. The optimization problem was formulated and solved, assuming that the channel coefficients associated with different RISs are independent and identically distributed (i.i.d.) Rayleigh random variables (RVs). Another optimization problem was also proposed in \cite{Mei} where the authors exploited the line-of-sight (LoS) link between nearby RISs to construct a multi-hop cascaded LoS link between the base station (BS) and user where a set of RISs are selected to successively reflect the BS's signal so that the received signal power at the user is maximized. However, the authors only considered the impact of path-loss and ignored the impact of fading in that study. A very recent paper on optimizing the performance of RIS-aided networks is the one presented in \cite{ZYang}. Multiple RISs were utilized to serve wireless users where the energy efficiency of the network was maximized by dynamically controlling the on-off status of each RIS as well as optimizing the reflection coefficients matrix of the RISs. More on optimization problems in RIS-aided networks with multiple RISs could be found in \cite{Lyu} and \cite{George}. 

One of the early papers on performance analysis of multiple RISs networks is \cite{Jung}. The authors relied in their uplink rate analysis on the law of large numbers and the CLT to approximate the distributions of RVs, which are functions of the squared magnitude of the channel coefficient, by Gaussian distributions. Again, the CLT analysis approach could be considered accurate only at a very high number of reflecting elements. Another work that utilizes the CLT analysis approach in its evaluation of the performance of RIS-aided networks is \cite{Gala}. The authors derived tight approximations/bounds for the OP, achievable rate, and average SER of RIS-aided networks in the presence of direct link and assuming Nakagami-$m$ fading channels. In \cite{LYang}, the authors provided a general expression for the OP and approximate asymptotic expression for the sum-rate of a multi-RIS network over Rayleigh fading channels. The derivations were based on two scenarios, one approximates the RIS channels using the non-central chi-square (NCCS) distribution, and the other approximates it using the $K_{G}$ distribution. The channel amplitudes were assumed to be i.i.d. RVs with fixed mean and variance. In \cite{Zhang}, the authors discussed the impact of centralized and distributed RIS deployment strategies on the capacity region of multi-RIS-aided systems. However, in the distributed RIS deployment, the authors assumed that the channels associated with the different distributed RISs undergo i.i.d. Rayleigh fading. 

A multi-RIS-aided system for both indoor and outdoor communications was considered in \cite{Yildirim}. Aiming for low-complexity transmission, the authors proposed a RIS selection strategy that selects the RIS with the highest SNR to assist the communication. However, small-scale fading was ignored, and the performance analysis of the RIS selection strategy was not carried out. In a very recent study, the authors proposed in \cite{Do} two schemes, one called exhaustive RIS-aided scheme and the other called opportunistic RIS-aided scheme. In the exhaustive RIS-aided scheme, all RISs help in sending the source message to the destination, whereas in the opportunistic RIS-aided scheme, only the RIS of the best SNR forwards the source message to the destination. The authors showed that either a Gamma distribution or a Log-Normal distribution can be used to approximate the distribution of the magnitude of the end-to-end (e2e) channel coefficients in both schemes. Tight approximate closed-form expressions were derived for the system OP and ergodic capacity assuming i.ni.d. fading channels scenario. Most recently, a study that investigates the error probability performance of a multi-RIS-aided network has been proposed in \cite{Phan}. Closed-form expression was derived for the system SEP assuming i.i.d. Nakagami-$m$ fading channels.	
	
For the sake of performance analysis simplicity, some existing works \cite{Mei} and \cite{Zhang}-\hspace{-0.01cm}\cite{Yildirim} only considered path-loss effects and/or ignored small-scale fading effects, which is not practical. On the other hand, some works relied on the deterministic fading channel \cite{Jung} or the i.i.d. fading channel model \cite{Gala} and \cite{LYang}. Nevertheless, this assumption is too optimistic and does not reflect the real channels in practice. This can be interpreted as in distributed multi-RIS-aided systems, the RISs are installed significantly far apart, e.g., tens of meters, and hence channels between different RISs cannot be assumed to be i.i.d.. It is worth mentioning that the characterization of distributions that be utilized to statistically model the magnitude of the e2e channel coefficient of distributed multi-RIS-aided systems is still a challenging problem. One can observe that some previous works, \cite{Liang} and \cite{Jung}, adopted the assumption that the channel distribution of RIS networks can be approximated by the NCCS where the CLT analysis approach could be used. However, results based on CLT are valid only for a very high number of reflecting or meta-surface elements. Besides, it has been noticed that the aforementioned studies do not provide expressions for the system diversity order and coding gains, which are very useful metrics that provide meaningful insights into the system performance. Furthermore, the aforementioned studies do not provide key results on the impact of RIS distribution (serial or parallel), random and equal reflecting elements distributions, and RIS locations on the system performance.
	
Motivated by the above observations, the contributions of this work can be summarized as follows:
	
\begin{itemize}
\item We provide a performance analysis of multi-RIS-aided networks over Nakagami-$m$ fading channels. For that purpose, we first prove that the RIS channel distributions can be modeled/treated using the Laguerre series method \cite{Primak}. In this paper, to cover more scenarios and to have flexibility in RIS locations between the source and destination, we consider some practical multi-RIS-aided system settings. Hence, the i.ni.d. case is considered for RIS channels. In addition, both small-scale fading and path-loss effects have been considered in the derivation of RIS channel statistics, including probability density function (PDF) and cumulative distribution function (CDF). 
		
\item Utilizing the achieved SNR statistics, we derive closed-form approximate expressions for the system OP and average symbol error probability (ASEP). The derived expressions are valid for an arbitrary number of reflecting elements and non-integer values of fading parameter $m$. 
		
\item To get more insights into the system performance, the derived expression of the OP is utilized in formulating an optimization problem where the optimum number of reflecting elements that guarantees a predetermined outage performance and the optimal relative placement of RISs are determined. In addition, a closed-form simple expression is derived for the asymptotic OP as well as the diversity order and coding gain. The provided results are unique and presented for the first time in this study.
		
\item The impacts of several system parameters, including the number of reflecting elements, number of RISs, distances between the source to RIS and RIS to destination and location of RISs, Nakagami-$m$ shaping and scaling parameters, etc. are investigated under realistic conditions as shown in Section \ref{Numerical Results}.

\end{itemize}
	
The rest of this paper is organized as follows. Section \ref{SCMs} presents the system and channel models. The performance analysis is evaluated in Section \ref{EPA}. Section \ref{Opt} gives the formulation and solution of the optimization problem. Some simulation and numerical results are discussed in Section \ref{Numerical Results}. Finally, the paper is concluded in Section \ref{C}.
	
\textit{Notations and symbols:}
It is worth mentioning that the notations and symbols used throughout this paper are, respectively given in Tables \ref{TEST1} and \ref{TEST2}.

\section{System and Channel Models}\label{SCMs}

In this section, we illustrate the proposed system and channel models of multi-RIS assisted networks. Then, the RIS selection strategy is presented.
\subsection{System and Channel Models}
As shown in Fig. \ref{System model}, we consider a multi-RIS assisted network, in which a single-antenna source (S) communicates a single-antenna destination (D) with the help of $K$ RISs. Each RIS, $\left\{\text{RIS}_{k}\right\}_{k=1}^{K}$, is equipped with $N_{k}$ passive elements. The S-RIS$_{k}$ links in the first hop are assumed to undergo Nakagami-\emph{m} fading, where the channel vector between the S and RIS$_{k}$ is denoted by $\mathbf{h}_{k}\in \mathbb{C}^{N_{k}\times1}$, $\mathbf{h}_{k} = \left[h^{(1)}_{k}, ..., h^{(i)}_{k}, ..., h^{(N_{k})}_{k}\right]^{T}$, $h^{(i)}_{k}=\frac{1}{\sqrt{P_{L,k,1}}}\alpha^{(i)}_{k}e^{-j\phi^{(i)}_{k}}$ denotes the channel coefficient between the S and RIS$_{k}$ $i$th element, where $P_{L,k,1}$, $\alpha^{(i)}_{k}$, and $\phi^{(i)}_{k}$ refer to the path-loss, channel amplitude, and channel phase, respectively for the first hop. Likewise, the RIS$_{k}$-D links are also assumed to have Nakagami-\emph{m} fading channel model, where the channel vector between the RIS$_{k}$ and D is denoted by $\mathbf{g}_{k}\in \mathbb{C}^{N_{k}\times1}$, $\mathbf{g}_{k} = \left[g^{(1)}_{k}, ..., g^{(i)}_{k}, ..., g^{(N_{k})}_{k}\right]^{T}$, $g^{(i)}_{k} = \frac{1}{\sqrt{P_{L,k,2}}}\beta^{(i)}_{k}e^{-j\Phi^{(i)}_{k}}$ denotes the channel coefficient between $i$th RIS element and D, where $P_{L,k,2}$, $\beta^{(i)}_{k}$, and $\Phi^{(i)}_{k}$ refer to the path-loss, channel amplitude, and channel phase, respectively for the second hop. The reflection coefficients of the RIS$_{k}$ are denoted by the entries of the diagonal matrix $\mathbf{\Theta}_{k}\in \mathbb{C}^{N_{k}\times N_{k}}$, for the $i$th element. Under the full reflection assumption, we have $\Theta^{(i,i)}_{k}=e^{j\theta^{(i)}_{k}}$, where $\theta^{(i)}_{k}\in[0,2\pi)$.

\begin{table}[t!]
	\centering
	\caption{Notations used in the paper.}
	\label{TEST1}
	\begin{tabular}{|l|l|}	
		\hline
		\textbf{Notation}          & \textbf{Definition}                                               \\ \hline
		$ P_{r}\left[\cdot\right] $           & Probability operator                                                      \\ \hline
		$F_{X}\left(x\right)$      & Cumulative distribution function (CDF) of a random variable $ X $ \\ \hline
		$f_{X}\left(x\right)$      & Probability density function (PDF) of a random variable $ X $     \\ \hline
		$|\cdot|$                  & Absolute value                                                    \\ \hline
		$\mathbb{C}^{m\times n}$  & Set of matrices with dimension $m\times n$                                      \\ \hline
		
		$(\cdot)^{T}$              & Transpose operator                                              \\ \hline
		${E}\left(\cdot\right)$               & Expectation operator                                              \\ \hline
		${Var}\left(\cdot\right)$               & Variance operator                                              \\ \hline
		$\gamma(\cdot,\cdot)$        & Lower incomplete Gamma function                                   \\ \hline
		$\Gamma(\cdot)$            & Gamma function                                                    \\ \hline
		$\exp(\cdot)$             & Exponential function    \\ \hline
		$ Q(\cdot) $             &   Q-function    \\ \hline
	\end{tabular}
\end{table}

\begin{table}[t]
	\centering
	\caption{Symbols used in the paper.}
	\label{TEST2}
	\begin{tabular}{|l|l|}	
		\hline
		\textbf{Symbol}          & \textbf{Definition}                                               \\ \hline
		$K$      & Number of RISs     \\ \hline
		$N_{k}$      & Number of reflecting elements of $k$th RIS (RIS$_{k} $)    \\ \hline
		$h^{(i)}_{k}$           & Channel coefficient between source (S) and RIS$_{k}$ $i$th element                                                      \\ \hline
		$\mathbf{h}_{k}$      & Channel vector between the S and RIS$_{k}$  \\ \hline
		$g^{(i)}_{k}$      & channel coefficient between $i$th RIS element and destination (D)     \\ \hline
		$\mathbf{g}_{k}$                  & Channel vector between the RIS$_{k}$ and D                                                   \\ \hline
		$P_{L,k,n}$                   & Path loss belongs to RIS$_{k}$ in $n$th hop                                                  \\ \hline
		$\alpha^{(i)}_{k}$                   & Channel amplitude belongs to RIS$_{k}$ in the first hop                                                  \\ \hline
		$\beta^{(i)}_{k}$                   & Channel amplitude belongs to RIS$_{k}$ in the second hop                                                  \\ \hline
		$\phi^{(i)}_{k}$                  & Channel phase belongs to RIS$_{k}$ in the first hop                                                  \\ \hline
		$\Phi^{(i)}_{k}$                   & Channel phase belongs to RIS$_{k}$ in the second hop                                                 \\ \hline
		$\mathbf{\Theta}_{k}$              & Diagonal matrix for reflection coefficients of the RIS$_{k}$                                              \\ \hline
		$E_{s}$               & Average power of the transmitted signal                                              \\ \hline
		$\lambda$ and $ f_{c} $               & Wavelength and carrier frequency, respectively                                            \\ \hline
		$ G_{k,n}$       & Gain of the RIS$_{k}$ in $n$th hop                                   \\ \hline
		$\epsilon_{k}$            & Efficiency of the RIS$_{k}$                                                   \\ \hline
		$ d_{k,1} $ & Distance from the source-to-RIS$_{k}$ \\ \hline
		$ d_{k,2} $ & Distance from the RIS$_{k}$-to-destination \\ \hline
		$\bar{\gamma}$ and $\gamma_\text{out}$           & Average SNR and predetermined outage threshold, respectively    \\ \hline
		$m_{k,n}$         & Shape parameter for Nakagami distribution of RIS$_{k}$ in $n$th hop   \\ \hline
		$\Omega_{k,n}$           & Spread parameter for Nakagami distribution of RIS$_{k}$ in $n$th hop   \\ \hline
		
	\end{tabular}
\end{table}

The signal received at D from the reflected signals of RIS$_{k}$ can be expressed as
\begin{align} 
\label{receivedsignal}
y_{k} &= \sqrt{\frac{E_{s}}{P_{L,k}}}\sum_{i = 1}^{N_{k}}\alpha^{(i)}_{k}\beta^{(i)}_{k}s + n_{k},
\end{align}
where $E_{s}$ is the average power of the transmitted signal, $s$ is the transmitted signal, ${n}_{k}$ denotes the additive white Gaussian noise (AWGN) sample with zero mean and variance $N_{0}$, and $P_{L,k}$ denotes the overall path-loss of the RIS$_{k}$-assisted path. Here, $ P_{L,k} = P_{L,k,1}P_{L,k,2} $ and given by
\begin{align} 
\label{pathloss}
P_{L,k} &= \left(\left(\frac{\lambda}{4\pi}\right)^{4}\frac{G_{k,1}G_{k,2}}{d^{2}_{k,1}d^{2}_{k,2}}\epsilon_{k}\right)^{-1},
\end{align}
where $\lambda$ is the wavelength, $ G_{k,1}$ and $G_{k,2}$ are the gains of the RIS$_{k}$ in the first and second hops, respectively, and $\epsilon_{k}$ is the efficiency of RIS$_{k}$, which is described as ratio of transmitted signal power by RIS to received signal power by RIS. In this paper, it is assumed that $\epsilon_{k} = 1$. Additionally, $ d_{k,1} $ and $ d_{k,2} $ represent the distances from the source-to-RIS$_{k}$ and RIS$_{k}$-to-destination, respectively.

\begin{figure}[t]
	\includegraphics[width=.4\textwidth]{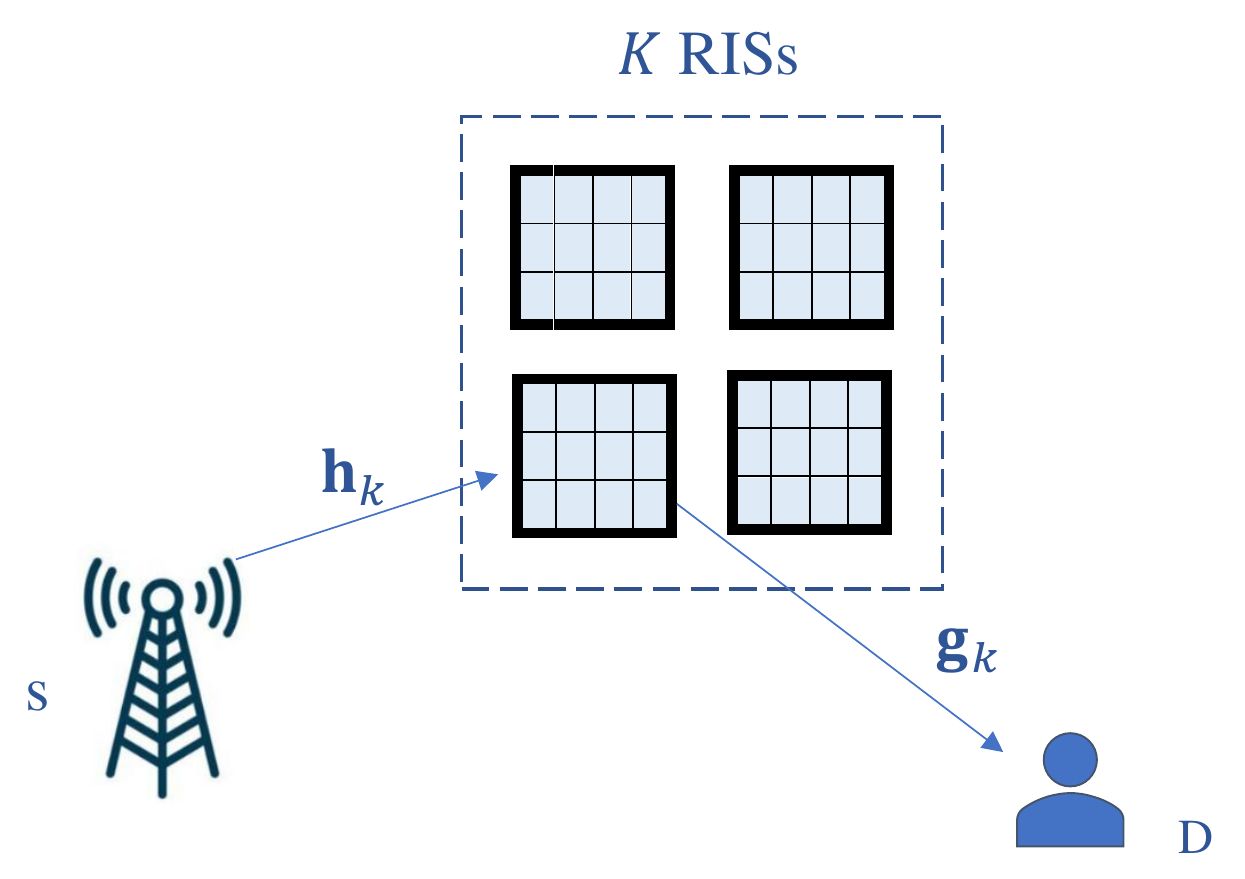}
	\centering
	\caption{A multi-RIS aided communication system model.}	
	\label{System model}
	\vspace{-0.5cm}
\end{figure}

The RIS$_{k}$ optimizes the phase reflection coefficients to maximize the received SNR at D, by aligning the phases of the reflected signals to the sum of the phases of its incoming and outgoing fading channels. Thus, the maximized e2e SNR for RIS$_{k}$ can be expressed as 
\begin{equation} \label{eq:gamma}
\gamma_{k} =\frac{E_{s}}{N_{0}P_{L,k}}\left ({\sum _{i=1}^{N_{k}} \alpha^{(i)}_{k}\beta^{(i)}_{k} }\right)^{2}= \frac{\bar{\gamma}}{P_{L,k}}Z^{2}_{k},
\end{equation}
where $\bar{\gamma}=\frac{E_{s}}{N_{0}}$ is the average SNR. 

{\color{black}{The PDF of the Nakagami-$m$ RV, charactarized by $m$ (shape parameter) and $\Omega$ (scale parameter), is given by
		$$
		\frac{2 m^m}{\Gamma(m) \Omega^m} x^{2m-1} \exp\left\{-\frac{m}{\Omega}x^2\right\}.
		$$	
		
		Based on the proposed distribution of $\alpha^{(i)}_{k}$ and $\beta^{(i)}_{k}$, the true distribution of $\gamma_{k}$ can well be approximated by the first term of a Laguerre series expansion, which results in gamma distribution, as proved in Theorem~\ref{thm:PDFgammak}.
		
		\begin{theorem}\label{thm:PDFgammak}
			The PDF of e2e SNR for RIS$_{k}$, $\gamma_{k}$, expressed in \eqref{eq:gamma}, is given by
			\begin{align}\label{eq:PDFgamma}
				f_{\gamma_{k}}(\gamma)\simeq\frac{\left(\frac{ P_{L,k} }{\bar{\gamma}}\gamma\right)^{\frac{a_{k}-1}{2}} \exp \left(-\sqrt{\frac{ P_{L,k}}{\bar{\gamma}b^{2}_{k}}\gamma}\right)}{2\Gamma (a_{k}) b^{a_{k}}_{k}  \sqrt{\frac{\bar{\gamma} }{P_{L,k}}\gamma }}.
			\end{align}
			Accordingly, the CDF of $\gamma_{k}$ is given by
			\begin{align}\label{eq:CDFgamma}
				F_{\gamma_{k}}(\gamma)\simeq \frac{\gamma \left(a_{k},\sqrt{\frac{P_{L,k} }{\bar{\gamma}b^{2}_{k}}\gamma}\right)}{\Gamma (a_{k})}.
			\end{align}
		\end{theorem}
		\begin{proof}
			Let $V^{(i)}_{k}=\alpha^{(i)}_{k}\beta^{(i)}_{k}$, where $\alpha^{(i)}_{k}$ has a Nakagami-$m$ distribution with shape parameter $m_{k,1}$ and scale parameter $\Omega_{k,1}$, and $\beta^{(i)}_{k}$ has a Nakagami-$m$ distribution with shape parameter $m_{k,2}$ and scale parameter $\Omega_{k,2}$. Also, it is assumed that $\alpha^{(i)}_{k}$ and $\beta^{(i)}_{k}$ are independent. Then, it can be shown that the PDF of $V^{(i)}_{k}$ is given by
			\begin{align}\label{eq:DoubleNakag}
				f_{V^{(i)}_{k}}(v)&=\frac{4(\Theta_1\Theta_2)^{m_{k,1}+m_{k,2}}}{\Gamma(m_{k,1})\Gamma(m_{k,2})} v^{m_{k,1}+m_{k,2}-1} \nonumber\\
				&\times K_{m_{k,1}-m_{k,2}}\left(2\Theta_1\Theta_2 v\right),
			\end{align}
			where $\Theta_j=\sqrt{\frac{m_{k,j}}{\Omega_{k,j}}}\,(j=1,2)$, and $K_{\nu}(\cdot)$ is the modified $\nu$-order Bessel function of the second kind \cite[Eq. (8.432)]{Grad.}.
			Now, let $Z_k=\sum_{i=1}^{N_{k}} V^{(i)}_{k}$. As $Z_k$ is the sum of i.i.d. positive RVs, and according to \cite[Sec. 2.2.2]{Primak}, the PDF of $Z_k$ can well be approximated by the first term of a Laguerre series expansion, which results in gamma distribution with PDF given by
			\begin{equation}\label{eq:PDFgamma}
				f_{Z_k}(z)\simeq\frac{1}{\Gamma{\left(a_k\right)}b_k^{a_k}}z^{a_k-1}\exp\left\{-\frac{z}{b_k}\right\},
			\end{equation}
			where the shape parameter $a_k=\frac{(E(Z_k))^2}{Var(Z_k)}$ and the scale parameter $b_k=\frac{Var(Z_k)}{E(Z_k)}$. As the expectation is a linear operator, $E(Z_k)=N_k E\left(V^{(i)}_{k}\right)$, and because of independence $Var(Z_k)=N_k Var\left(V^{(i)}_{k}\right)$. It can be shown that the $r^{th}$ moment of $V^{(i)}_{k}$, \cite[Eq. (5)]{Shankar2004}, is given by
			\begin{equation}\label{eq:rthmomentV}
				E(V^{(i)r}_{k})=\frac{\Gamma \left(m_{k,1}+\frac{r}{2}\right) \Gamma \left(m_{k,2}+\frac{r}{2}\right)}{(\Theta_1\Theta_2)^r\Gamma (m_{k,1}) \Gamma (m_{k,2})}.
			\end{equation}
			Finally, as $Z_k$ is a positive RV, the PDF of $\gamma_{k}$ can be obtained by applying the transformation of RVs technique, and this completes the proof.
		\end{proof}
		
		\subsection{RIS Selection Strategy}\label{sec:RIS_sel_st}
		Using all the RISs will increase the overall system complexity. Nevertheless, utilizing RIS selection over multiple RISs will provide a low-complexity and cost-effective transmission while many advantages of multiple RIS-aided systems are still preserved. In this system, we consider that one out of $K$ RISs is selected to aid the communications. Specifically, an RIS selection is carried out over $K$ RISs to select the RIS, which has maximum SNR. Accordingly, the maximized e2e SNR of the selected RIS can be expressed as
		\begin{equation}\label{eq:RISselection}
			\gamma^{*}=\underset{k= 1,...,K}{{\max}}\{\gamma_{k}\}. 
		\end{equation}

It is worth mentioning that since RISs are passive, the source has to estimate all channels and determine the RIS with the highest SNR. In order to realize \eqref{eq:RISselection}, we assume that the source has the necessary knowledge of the channel-state-information (CSI). We also assume a time-division duplexing (TDD) system. During the training period, assuming that uplink and downlink channels are reciprocal, the destination sends pilot symbols to the source. Next, the source estimates the channel coefficients for each RIS and then specifies RIS with maximum SNR \cite{Alwazani2020}.

		\begin{theorem}
			The CDF and PDF of $\gamma^{*}$, expressed in \eqref{eq:RISselection} can be, respectively given by
			\begin{align}\label{eq:CDFRISselectionsimple}
				F_{\gamma^{*}}(\gamma)& = \sum_{n_{1} = 0}^{\infty}...\sum_{n_{K} = 0}^{\infty} \prod_{k=1}^{K}\frac{\left(-1\right)^{n_{k}}\left(\sqrt{\frac{P_{L,k} }{\bar{\gamma}b^{2}_{k}}}\right)^{a_{k} + n_{k}}}{n_{k}!\left(a_{k} + n_{k}\right)\Gamma (a_{k})}\nonumber\\
				&\times \gamma^{\frac{\sum_{k=1}^{K}\left(a_{k} + n_{k}\right)}{2}},
			\end{align}
			and 
			{\footnotesize\begin{align}\label{eq:PDFRISselection}
					f_{\gamma^{*}}(\gamma) &= \sum_{j = 1}^{K}\frac{\left(\frac{ P_{L,j} }{\bar{\gamma}}\gamma\right)^{\frac{a_{j}-1}{2}} \exp \left(-\sqrt{\frac{ P_{L,j}}{\bar{\gamma}b^{2}_{j}}\gamma}\right)}{2 b^{a_{j}}_{j} \Gamma (a_{j}) \sqrt{\frac{\bar{\gamma} }{P_{L,j}}\gamma }}\sum_{k\ne j, k=1}^{\infty} ... \sum_{k\ne j, k=K}^{\infty}\nonumber\\
					& \prod_{k\ne j, k=1}^{K}\frac{\left(-1\right)^{n_{k}}\left(\sqrt{\frac{P_{L,k} }{\bar{\gamma}b^{2}_{k}}}\right)^{a_{k} + n_{k}}}{n_{k}!\left(a_{k} + n_{k}\right)\Gamma (a_{k})}\times \gamma^{\frac{\sum_{k\ne j, k=1}^{K}\left(a_{k} + n_{k}\right)}{2}}.
			\end{align}}
		\end{theorem}
		
		\begin{proof} 
			According to the order statistics theory, the CDF of $\gamma^{*}$ can be formulated as
			\begin{equation}\label{eq:CDFRISselection}
				F_{\gamma^{*}}(\gamma) = \prod_{k = 1}^{K}F_{\gamma_{k}}(\gamma). 
			\end{equation}
			With the help of $\gamma\left(\alpha, x \right) = \sum_{n = 0}^{\infty} \frac{\left(-1\right)^{n}x^{\alpha+n}}{n!\left(\alpha + n\right)}$ \cite[Eq. (8.354.1)]{Grad.}, the lower incomplete Gamma function in \eqref{eq:CDFgamma} can be expanded into series. With some mathematical simplifications, we can obtain the CDF of $\gamma^{*}$ as in \eqref{eq:CDFRISselectionsimple}. The PDF of $\gamma^{*}$ can be obtained by taking the derivative of \eqref{eq:CDFRISselection}, $ f_{\gamma^{*}}(\gamma) = \frac{d }{d x}F_{\gamma^{*}}(\gamma)= \frac{d}{d x}\prod_{k = 1}^{K}F_{\gamma_{k}}(\gamma)=\sum_{j = 1}^{K}f_{\gamma_{j}}(\gamma)\prod_{k\ne j, k=1}^{K}F_{\gamma_{k}}(\gamma) $. Substituting \eqref{eq:PDFgamma} and \eqref{eq:CDFgamma}, then we achieve the PDF of $\gamma^{*}$ as in \eqref{eq:PDFRISselection} and then the proof is completed.
		\end{proof}
}}

\section{Performance Analysis}\label{EPA}
In this section, we derive theoretical expressions for the system OP and ASEP. 
\subsection{Outage Probability}
The derivation of system OP is summarized in the following theorem and corollary.
\begin{theorem}
The OP can be given by
\begin{equation}\label{eq:Pout}
P_{out}=F_{\gamma^{*}}(\gamma_\text{out}).
\end{equation}
\end{theorem}
\begin{proof} 
The system OP is defined as the event where the e2e SNR goes below a predetermined outage threshold $\gamma_\text{out}$. Mathematically speaking $P_{out}=\mathrm{Pr}\left[\gamma\leq\gamma_{th}\right]$,
where $\mathrm{Pr}[.]$ is the probability operation. Thus, using (\ref{eq:CDFRISselectionsimple}), we have \eqref{eq:Pout} and then the proof is completed.
\end{proof}

In order to gain further insights into the OP performance, we give the following corollary.
\begin{corollary} Considering the asymptotic behavior of OP, for $\bar{\gamma}\rightarrow \infty$, the OP is given by
\begin{equation}\label{eq:AsympOut}
P_{ out}^{\infty} = \prod_{k = 1}^{K}\left(\frac{b^{2}_{k}}{\gamma_\text{out}\left(a_{k}!\right)^{-\frac{2}{a_{k}}}P_{L,k}}\bar{\gamma}\right)^{-\frac{a_{k}}{2}}.
\end{equation}
\end{corollary}
\begin{proof}
The asymptotic OP can be calculated by $P^{\infty}_{out}=F^{\infty}_{\gamma^{*}}(\gamma_\text{out})=\prod_{k = 1}^{K}F^{\infty}_{\gamma_{k}}(\gamma_\text{out})$. Then, using \cite[Eq. (8.354.1)]{Grad.}, we get
$F^{\infty}_{\gamma_{k}}(\gamma_\text{out})= \sum_{n_{k}=0}^{\infty}\frac{(-1)^{n_{k}}\left(\sqrt{\frac{\gamma_\text{out}}{\frac{\bar{\gamma}}{P_{L,k}}}}\right)^{a_{k}+n_{k}}}{(a_{k}+n_{k})b^{a_{k}+n_{k}}_{k}\Gamma(a_{k})} $. As $\bar{\gamma}\rightarrow\infty$, this expression is only dominated by the first term in summation. Upon considering that, we get \eqref{eq:AsympOut}, and then the proof is completed.
\end{proof}

In the very high average SNR region, the OP can expressed as $P_{ out}^{\infty}=(G_{c}\bar{\gamma})^{-G_{ d}}$, where $G_{ d}$ is the gain in diversity and $G_{ c}$ is the gain in coding \cite{Alouinib}. From Corollary 1, in the case of i.i.d., i.e., $a = a_{k}$, $b = b_{k}$, and $ P_{L} = P_{L,k} $, $ P_{ out}^{\infty} = \left(\left(\frac{b^{2}}{\gamma_\text{out}\left(a!\right)^{-\frac{2}{a}}P_{L}}\right)^{\frac{1}{K}}\bar{\gamma}\right)^{-\frac{a}{2}K} $. Accordingly, the coding gain is $G_{c}=\left(\frac{b^{2}}{\gamma_\text{out}(a!)^{-\frac{2}{a}}P_{L}}\right)^{\frac{1}{K}}$ and the diversity order is $G_{d}=\left(\frac{a}{2}\right)K$. 

\subsection{Average Symbol Error Probability}

\begin{theorem}
The ASEP can be given by
\begin{align}\label{eq:ASEPclosedform}
P_{e}&=\frac{p\sqrt{q}}{2\sqrt{\pi}}\sum_{n_{1} = 0}^{\infty}...\sum_{n_{K} = 0}^{\infty} \prod_{k=1}^{K}\frac{\left(-1\right)^{n_{k}}\left(\sqrt{\frac{P_{L,k} }{\bar{\gamma}b^{2}_{k}}}\right)^{a_{k} + n_{k}}}{n_{k}!\left(a_{k} + n_{k}\right)\Gamma (a_{k})}\nonumber\\
&\times \frac{1}{q^{\frac{\sum_{k=1}^{K}\left(a_{k} + n_{k}\right)}{2} + \frac{1}{2}}}\Gamma\left(\frac{\sum_{k=1}^{K}\left(a_{k} + n_{k}\right)}{2}+\frac{1}{2}\right),
\end{align}
where $p$ and $q$ are constants representing the type of modulation. 
\end{theorem}
\begin{proof} 
Our results apply for all general modulation formats that have an ASEP expression of the form	
\begin{align}\label{eq:ASEPgeneral}
P_{e}=E_{\gamma^{*}}\left[pQ(\sqrt{2q\gamma})\right],
\end{align} 
where $ Q(\cdot) $ is the Gaussian Q-function, and $ p $ and $ q $ are modulation-specific constants. Such modulation formats include binary phase-shift keying (BPSK) ($ p = 1, q = 1 $) and M-ary PSK ($ p = 2, q = \sin\left(\frac{2\pi}{M}\right) $).		
The ASEP can be obtained as follows \cite{McKay}
\begin{align}\label{eq:ASEP}
P_{e}=\frac{p\sqrt{q}}{2\sqrt{\pi}}\int_{0}^{\infty}\frac{\exp\left(-q\gamma\right)}{\sqrt{\gamma}}F_{\gamma^{*}}(\gamma) d\gamma.
\end{align} 
Upon inserting \eqref{eq:CDFRISselectionsimple} in \eqref{eq:ASEP}, then 
\begin{align}\label{eq:ASEP1}
P_{e}&=\frac{p\sqrt{q}}{2\sqrt{\pi}}\sum_{n_{1} = 0}^{\infty}...\sum_{n_{K} = 0}^{\infty} \prod_{k=1}^{K}\frac{\left(-1\right)^{n_{k}}\left(\sqrt{\frac{P_{L,k} }{\bar{\gamma}b^{2}_{k}}}\right)^{a_{k} + n_{k}}}{n_{k}!\left(a_{k} + n_{k}\right)\Gamma (a_{k})}\nonumber\\
&\times\int_{0}^{\infty} \gamma^{\frac{\sum_{k=1}^{K}\left(a_{k} + n_{k}\right)}{2} +\frac{1}{2}} \exp\left(-q\gamma\right) d\gamma.
\end{align}
With the help of $\int_{0}^{\infty}x^{v-1}\exp\left(-\mu x\right) dx = \frac{1}{\mu^{v}}\Gamma\left(v\right)$ \cite[Eq. (3.381.4)]{Grad.}, we obtain \eqref{eq:ASEPclosedform} and then the proof is completed.
\end{proof}

{\color{black}
\section{Optimization}\label{Opt} 

In this section, we aim to provide possible solutions to the following optimization questions:
\begin{enumerate}
\item Is it possible to determine the optimal number of reflecting elements on each RIS in order to guarantee a predetermined outage performance $P_\text{out}^\text{th}$?
\item What is the solution to the above problem if the objective is to minimize the total number of reflecting elements?
\item What is the optimal relative placement of RISs, S, and D, that will result in the minimum $P_\text{out}$?
\end{enumerate}

For optimization purpose, it is quite challenging to use the asymptotic OP in \eqref{eq:AsympOut}. Hence, we use the following asymptotic upper bound (UB) on the OP
\begin{equation}\label{eq:AsympOut_UB}
\widetilde{P}_\text{out} = \prod_{k = 1}^{K} \frac{e^{a_k} \left( \sqrt{ \frac{\gamma_\text{out} P_{L,k} } { \bar{\gamma} b^2_k } } \right)^{a_k}}{a_k^{a_k}},
\end{equation}
which is true because 
\[
\frac{\gamma \left(a_{k},\sqrt{\frac{P_{L,k} }{\bar{\gamma}b^{2}_{k}}\gamma_\text{out}}\right)}{\Gamma (a_{k})} \le  \frac{e^{a_k} \left( \sqrt{ \frac{\gamma_\text{out} P_{L,k} } { \bar{\gamma} b^2_k } } \right)^{a_k}}{a_k^{a_k}}.
\]
We can express the UB-OP in terms of the number of RIS elements $\mathbf{N} \triangleq \{N_1,\dots,N_K\}$ as
\begin{equation}\label{eq:AsympOut_UB_N}
\widetilde{P}_\text{out}(\mathbf{N}) \triangleq \prod_{k = 1}^{K} \frac{e^{N_k c_k} x_k ^{N_k c_k}}{ \left(N_k c_k \right)^{N_k c_k}},
\end{equation}
where $c_k \triangleq \frac{m_{k,1} m_{k,2} \Gamma(m_{k,1})^2 \Gamma(m_{k,2})^2}{ m_{k,1} m_{k,2} \Gamma(m_{k,1})^2 \Gamma(m_{k,2})^2 - \Gamma(m_{k,1} + \frac{1}{2} )^2 \Gamma(m_{k,2} + \frac{1}{2} )^2 } - 1$ and $ x_k \triangleq \sqrt{ \frac{\gamma_\text{out} P_{L,k} } { \bar{\gamma} b^2_k } }$.

Let us start with the first optimization problem, which can be formulated as the following feasibility problem
\begin{subequations}\label{p1}
\begin{eqnarray}
\mbox{find}  \ \  \mathbf{N}
\\
\mbox{s.t.}\quad N_k \le N_\text{max}, \ \forall \ k,  \label{N_const} \\
\widetilde{P}_\text{out}(\mathbf{N}) = P_\text{out}^\text{th}, \label{out_const} \\
N_k \in \mathbb{Z}^+, \ \forall \ k, \label{integ_const}
\end{eqnarray}
\end{subequations}
where we can relax the non-convex integer constraint \eqref{integ_const}, because the solution can be obtained by finding the ceil-function of optimized $\mathbf{N}$. In order to deal with the non-convex constraint \eqref{out_const}, we first take natural log of both sides of \eqref{out_const} because log is a monotonically increasing function of its argument. Therefore, problem \eqref{p1} can be formulated as
\begin{subequations}\label{p2}
\begin{eqnarray}
\mbox{find}  \ \  \mathbf{N} \quad
\mbox{s.t.} \quad \eqref{N_const}, \\
f(\mathbf{N}) = \ln P_\text{out}^\text{th},  \label{log_out_const}
\end{eqnarray}
\end{subequations}
where we achieve \eqref{log_out_const} because 
\begin{align}
\ln  \widetilde{P}_\text{out}(\mathbf{N}) &= \sum_{k=1}^K N_k c_k + N_k c_k \ln x_k - N_k c_k \ln (N_k c_k) \notag \\
&\triangleq f(\mathbf{N}).
\end{align}
Note that $ f(\mathbf{N})$ is concave in $\mathbf{N}$. This is because $N_k c_k \ln (N_k c_k)$ is a convex function (second-order derivative of $N_k c_k \ln (N_k c_k)$ is positive) and sum of concave functions is also concave. However, the equality constraint \eqref{log_out_const} is still non-convex, which makes the problem \eqref{p2} non-convex. Therefore, we solve \eqref{p2} by proposing a path-following algorithm, which generates a sequence of improved feasible points for \eqref{p2} and finally converges to a locally optimal solution. 

Let $\mathbf{N}^{(\kappa)} \triangleq \{N^{(\kappa)} _1,\dots,N^{(\kappa)} _K\} $ be a feasible point for \eqref{p2}, which is found from $(\kappa-1)$-th iteration. We need to approximate \eqref{log_out_const} with an affine function, which is given by
\begin{align}
 f(\mathbf{N}) &\le  f^{(\kappa)}(\mathbf{N})  \notag \\
      &=  \sum_{k=1}^K N^{(\kappa)} _k c_k + c_k \ln \left( \frac{x_k}{N^{(\kappa)}_k c_k} \right) N_k.
\end{align}
Therefore, we solve the following convex problem to generate a next feasible point $\mathbf{N}^{(\kappa+1)}$ for \eqref{p2}
\begin{subequations}\label{p3}
\begin{eqnarray}
\mbox{find}  \ \  \mathbf{N} \quad
\mbox{s.t.} \quad \eqref{N_const}, \\
f^{(\kappa)} (\mathbf{N}) = \ln P_\text{out}^\text{th}.  \label{log_out_const_app}
\end{eqnarray}
\end{subequations}
We will shortly outline the pseudo-code to solve the problem \eqref{p1}. 

Let us now formulate the second problem mentioned at the start of Section \ref{Opt}, which solves the problem \eqref{p1} with the objective of minimizing the total number of reflecting elements. The formulation is given below:
\begin{subequations}\label{p4}
\begin{align} 
\min_{\mathbf{N}} \ \ \sum_{k=1}^K N_k  \\
 \mbox{s.t.}\quad  \eqref{N_const},\eqref{out_const},\eqref{integ_const}.
\end{align}
\end{subequations}
By following the similar steps as outlined above in \eqref{p2}-\eqref{p3}, we can iteratively solve the following problem at $\kappa$-th iteration to generate a next feasible point $\mathbf{N}^{(\kappa+1)}$ for \eqref{p4}
\begin{align} \label{p5}
\min_{\mathbf{N}} \ \ \sum_{k=1}^K N_k  \quad
 \mbox{s.t.}\quad  \eqref{N_const},\eqref{log_out_const_app}.
\end{align}
The objective of solving the above problem is to figure out whether we can save the number of reflecting elements by using a single RIS or by using multiple RISs, which will be demonstrated through simulations in Section \ref{Numerical Results}. The path-following procedure to solve the problems \eqref{p1}/\eqref{p4} is provided in Algorithm \ref{alg1}.

\begin{algorithm}[t]
	\caption{Algorithm for solving optimization problem \eqref{p1}/\eqref{p4} } \label{alg1}
	\begin{algorithmic}[1]
		\State \textbf{Initialization}: Initialize a feasible point
$\mathbf{N}^{(0)}$.  Set $\kappa=0$.
		\State\textbf{Repeat until convergence of $N_k$ $\forall$ $k$. }:
		Generate $\mathbf{N}^{(\kappa+1)}$ by solving \eqref{p3}/\eqref{p5}. Reset $\kappa\leftarrow \kappa+1$.
		 \State \textbf{Output} $\mathbf{N}^{(\kappa)}$ as the optimal solution of \eqref{p1}/\eqref{p4}.
	\end{algorithmic}
\end{algorithm}

Finally, let us now formulate the third problem mentioned at the start of Section \ref{Opt}, which is about finding the optimal relative placement of RISs, S, and D with the objective of minimizing the $P_\text{out}$. This can be formulated as follows
\begin{subequations}\label{p6}
\begin{eqnarray}
\min_{\mathbf{d}}   \ \  \widetilde{P}_\text{out}(\mathbf{d})
\\
\mbox{s.t.}\quad 0 \le d_{k,1} \le D, \ \forall \ k,  \label{d_const} 
\end{eqnarray}
\end{subequations}
where $\mathbf{d} \triangleq \{d_{1,1},\dots d_{K,1}\}$ is a vector of distances between S and $K$ RISs, $D$ is the distance between S and D, and $\widetilde{P}_\text{out}(\mathbf{d})$ is the asymptotic UB on the OP, which is given by
\begin{align}
\widetilde{P}_\text{out}(\mathbf{d}) \triangleq  \prod_{k=1}^K z_k \left(d_{k,1} (D - d_{k,1}) \right)^a_k,
\end{align}
where $z_k \triangleq \frac{e^{a_k} \left(  \frac{\gamma_\text{out}}{\bar{\gamma} b_k^2} \right)^{a_k/2} (4 \pi)^{2a_k}} { ( \lambda^4 G_{k,1} G_{k,2} \epsilon_k )^{a_k/2} a_k^{a_k} }$. The formulation of \eqref{p6} assumes linear placement of RISs between S and D because $d_{k,2} = D - d_{k,1}$, however, this assumption is just for convenience purpose as our target is only to find the optimal relative placement of RISs, S, and D. 

The problem \eqref{p6} is non-convex due to non-convex objective function. In order to deal with it, we first take log of the OP as log is a monotonically increasing function of its argument. Therefore, the problem \eqref{p6} can be formulated as
\begin{eqnarray} \label{p7}
\min_{\mathbf{d}}   \ \ g(\mathbf{d}) \quad 
\mbox{s.t.} \quad \eqref{d_const},
\end{eqnarray}
where
\begin{align} 
g(\mathbf{d}) &\triangleq  \ln  \widetilde{P}_\text{out}(\mathbf{d}) \notag \\
 &= \sum_{k=1}^K \ln z_k + a_k \left(\ln d_{k,1} + \ln ( D - d_{k,1} ) \right) .
\end{align}
Note that $g(\mathbf{d})$ is concave in $\mathbf{d}$ because $D \ge d_{k,1}$ $\forall$ $k$. This implies that the objective function in \eqref{p7} is the minimization of a concave function, which requires us to only evaluate the OP, $\widetilde{P}_\text{out}(\mathbf{d})$, at the extreme points in the set, i.e., $\mathbf{d} = \left\{ \{0,0,\dots,0\},  \{D,0,\dots,0\} , \hdots, \{D,D,\dots,D\} \right\}$. We can also verify this by linearizing the objective function about $\mathbf{d}^{(\kappa)}$ at the $\kappa$-th iteration and iteratively solving the following problem at the $\kappa$-th iteration to generate a next feasible point $\mathbf{d}^{(\kappa+1)}$ for \eqref{p7}
\begin{eqnarray} \label{p8}
\min_{\mathbf{d}}   \ \ g^{(\kappa)}(\mathbf{d}) \quad 
\mbox{s.t.} \quad \eqref{d_const},
\end{eqnarray}
where 
\begin{align}
 g(\mathbf{d}) &\le  g^{(\kappa)}(\mathbf{d})  \notag \\
      &=  \sum_{k=1}^K \ln z_k + a_k \left(\ln d^{(\kappa)}_{k,1} + \ln ( D - d^{(\kappa)}_{k,1} ) \right) \nonumber\\
      & + a_k \frac{2 d^{(\kappa)}_{k,1} - D }{D - d^{(\kappa)}_{k,1} }  + a_k  \frac{D - 2 d^{(\kappa)}_{k,1}  }{ d^{(\kappa)}_{k,1} (D - d_{k,1}^{(\kappa)}) }   d_{k,1}.
\end{align}
}

\section{Numerical Results}
\label{Numerical Results}
In this section, the numerical results corresponding to the considered multiple RISs system are verified by Monte-Carlo simulations and used to validate the theoretical analyses. Unless otherwise stated, in all presented illustrations, the carrier frequency is assumed to be $f_{c} = 2.4$ GHz, the gains of the RISs in the first and second hops are, respectively given as $ G_{k,1} = G_{k,2} = 5$ dB, and $\gamma_\text{out} = 0$ dB.

In Fig. \ref{Pout_different_K_N}, we validate the achieved analytical and asymptotic results via comparing them with simulation results. It is obvious from this figure that the analytical results match the simulation ones. In addition, it can be seen that the asymptotic results converge to the analytical and simulation results at very high values of SNR. The figure clearly shows that both the number of reflecting elements ($N$) and number of RISs ($K$) affect the system diversity order ($G_{d}$), which coincides with the results achieved in Section \ref{EPA}. Clearly, as either $N$ or $K$ increases, the higher the diversity order ($G_{d}$) and the better the achieved performance.

\begin{figure}[]
	\includegraphics[scale=0.5]{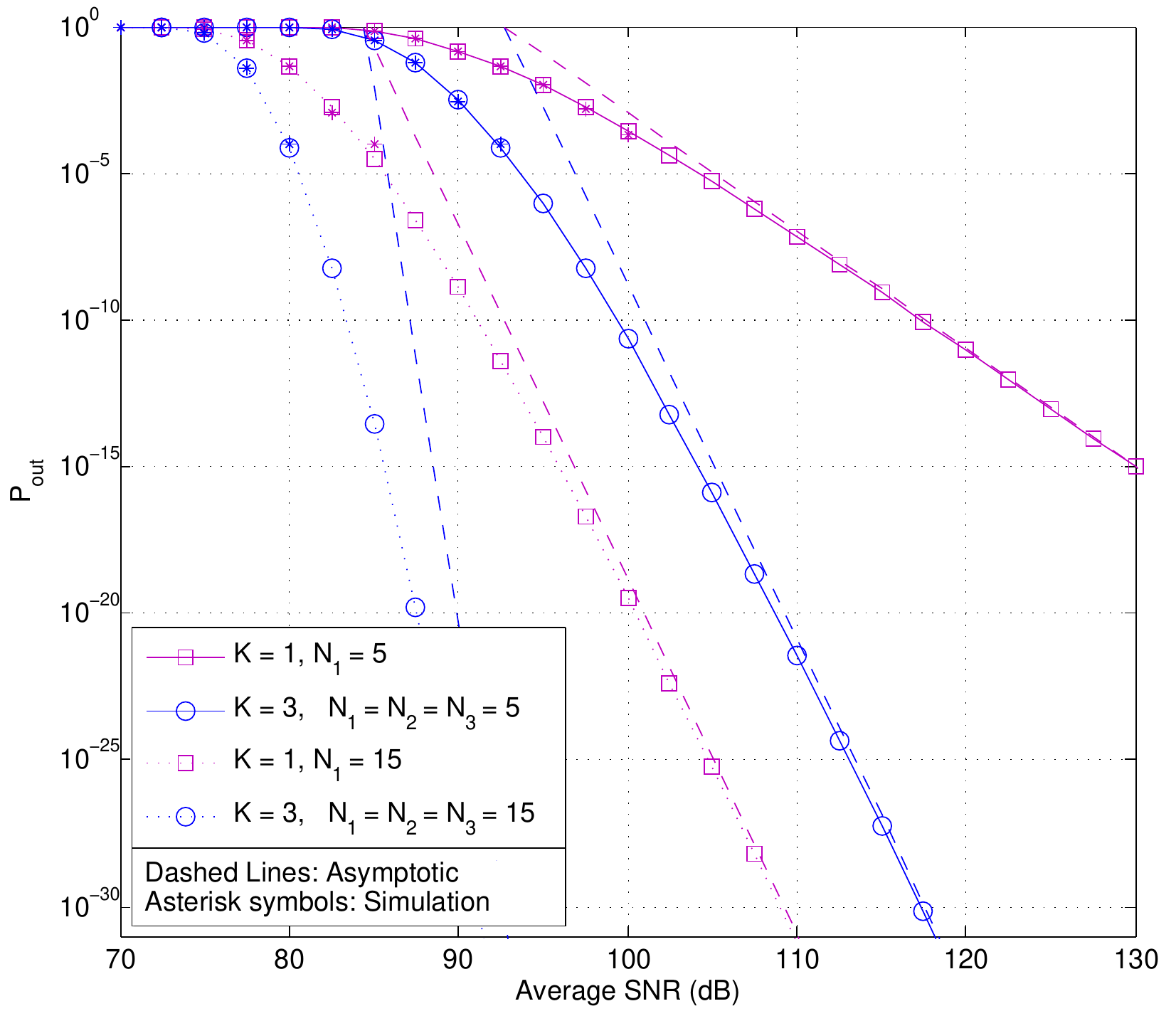}
	\centering
	\caption{The outage probability versus average SNR with different numbers of RIS $K\in\left\{1, 3\right\}$ and reflecting RIS elements $N_{1}, N_{2}, N_{3}\in\left\{5, 15\right\}$. $\left(m_{k,1}, m_{k,2}\right) = \left(1, 1\right)$, $\left(\Omega_{k,1}, \Omega_{k,2}\right) = \left(1, 1\right)$ and $\left(d_{k,1}, d_{k,2}\right) = \left(5 \hspace{0.15cm}\text{m}, 5 \hspace{0.15cm}\text{m}\right)$.}	
	\label{Pout_different_K_N}
\end{figure}

\begin{figure}[]
	\includegraphics[scale=0.5]{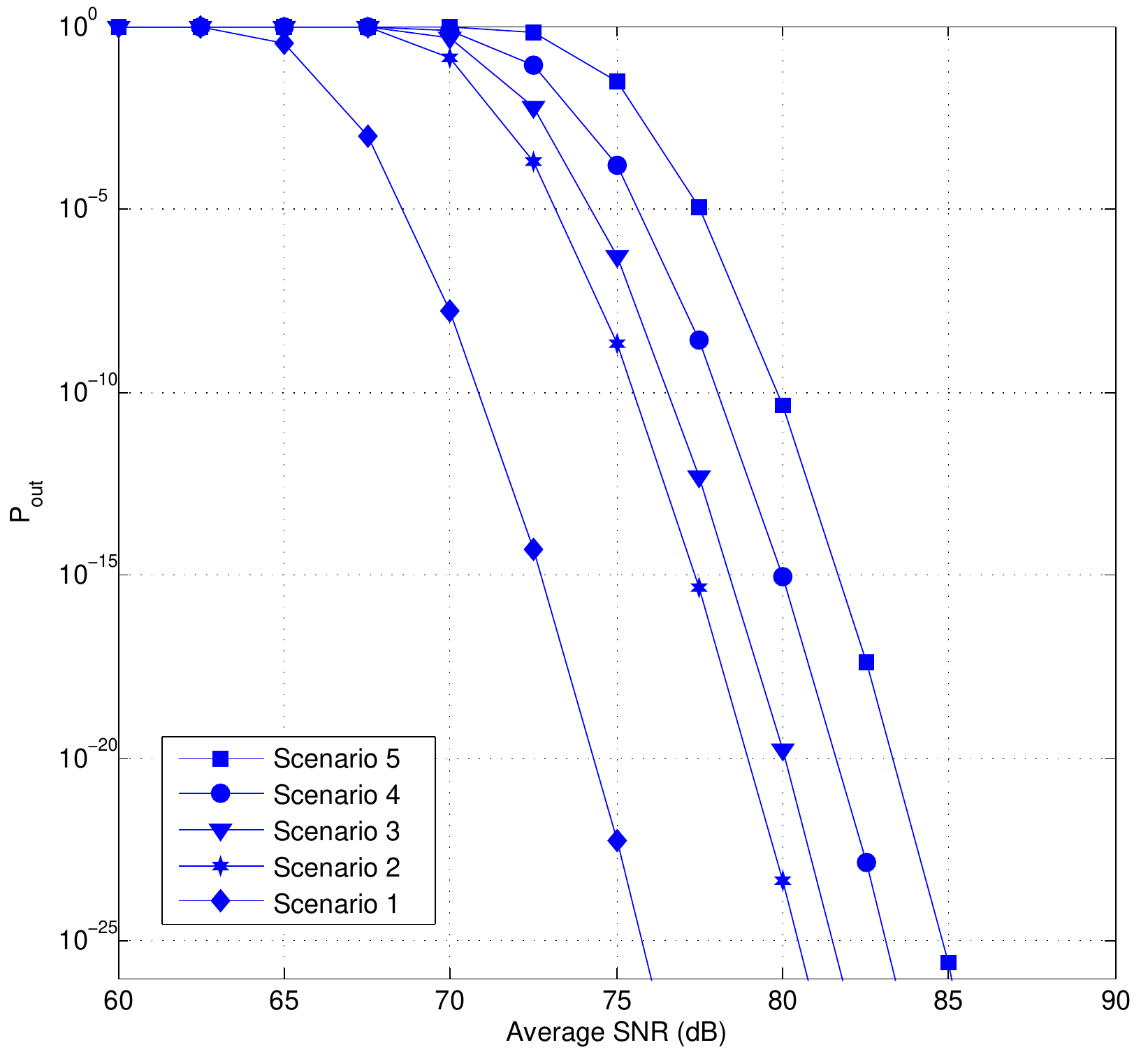}
	\centering
	\caption{The outage probability versus average SNR for five scenarios considering the impact of number of reflecting elements, number of RISs, and distances between source to RISs and between RISs to destination. $\left(m_{k,1}, m_{k,2}\right) = \left(1, 1\right)$ and $\left(\Omega_{k,1}, \Omega_{k,2}\right) = \left(1, 1\right)$.}	
	\label{Pout_1_16_2_30_3_20}
\end{figure}

To get more insights into the derived expressions, five important practical scenarios are studied in Fig. \ref{Pout_1_16_2_30_3_20}. The aim of this figure is to study the joint impact of number of reflecting elements ($N$), number of RISs ($K$), and distances between source to RISs and between RISs to destination on the system performance. In this figure, the distance between the source and destination is kept as 10 m and the total number of reflecting elements is also kept as 60. In the first scenario, we have one RIS with $N_{1}=60$. In the second scenario, we have two RISs serially distributed between the source and destination with $N_{1}=N_{2}=30$. The third scenario is a modified version of the second one where here the two RISs are distributed in parallel between the source and destination. In the fourth scenario, we have three RISs serially distributed between the source and destination with $N_{1}=N_{2}=N_{3}=20$. The fifth scenario is a modified version of the fourth one where the three RISs are distributed in parallel between the source and destination. In Scenario 1, we have $d_{1,1}=d_{1,2}=5$ m. In Scenario 2, we have $d_{1,1}=3.3$ m, $d_{1,2}=6.7$ m and $d_{2,1}=6.7$ m, $d_{2,2}=3.3$ m, whereas in Scenario 3, we have $d_{1,1}=d_{1,2}=5$ m and $d_{2,1}=d_{2,2}=5$ m. In Scenario 4, we have $d_{1,1}=2.5$ m, $d_{1,2}=7.5$ m and $d_{2,1}=5$ m, $d_{2,2}=5$ m and $d_{3,1}=7.5$ m, $d_{3,2}=2.5$ m, whereas in Scenario 5, we have $d_{1,1}=d_{1,2}=5$ m and $d_{2,1}=d_{2,2}=5$ m and $d_{3,1}=d_{3,2}=5$ m. It is obvious from this figure that the diversity order ($G_{d}$) for the five scenarios is equal as the five curves have the same slope. We can also see that Scenario 1 is outperforming the other four scenarios in terms of system coding gain ($G_{c}$). This is a key result and being reported for the first time in this paper where we can see that assigning a fixed number of reflecting elements to one RIS gives better performance than distributing them on multiple RISs. The more the number of RISs, the worse the achieved performance. This is because both the number of RISs and number of reflecting elements also have effect on the system coding gain and not only the diversity order, which is fixed in all these scenarios. Now, comparing the scenario where the RISs are serially distributed between the source and destination with the scenario where they are distributed in parallel, one can see that Scenarios 2 and 4 outperform Scenarios 3 and 5 for both two and three RISs cases. This is because the distance between communicating nodes directly affects the path-loss of the system, and hence, affecting the coding gain. 

\begin{figure}[]
	\includegraphics[scale=0.5]{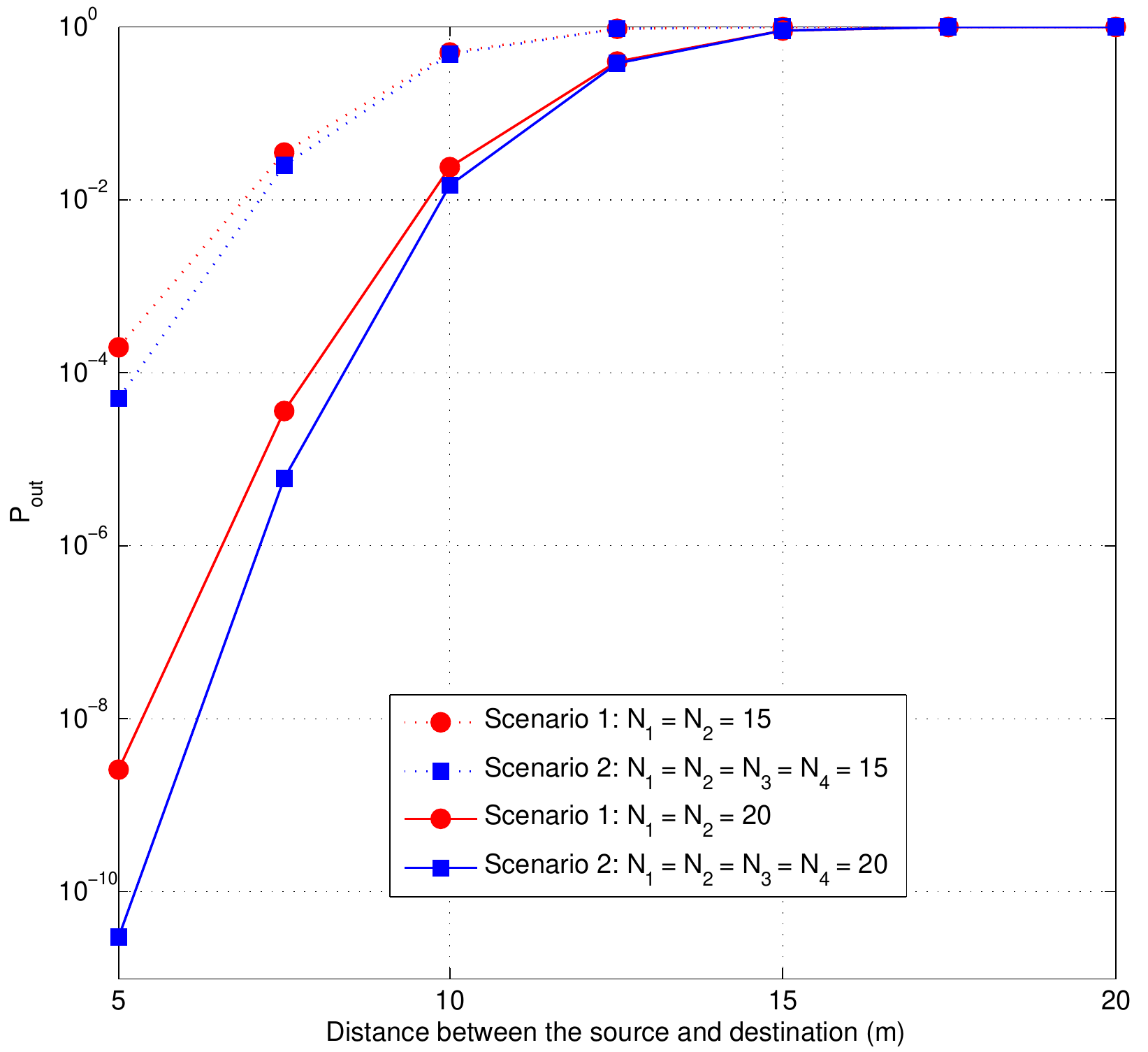}
	\centering
	\caption{The outage probability versus source-to-destination distance for two scenarios: one with two RISs and the other with four RISs. $K\in\left\{2, 4\right\}$, $N_{1}, N_{2}, N_{3}, N_{4}\in\left\{15, 20\right\}$, $\left(m_{k,1}, m_{k,2}\right) = \left(1, 1\right)$ and $\left(\Omega_{k,1}, \Omega_{k,2}\right) = \left(1, 1\right)$.}	
	\label{Pout_distance}
\end{figure}

Two scenarios that investigate the effect of distance between the source and destination on the system performance are portrayed in Fig. \ref{Pout_distance}. In this figure, the source is assumed to be located at a center of a circle of radius 5 m, and the destination is assumed to be located somewhere at the circumference of the circle. Then, we assume the source moving away from the source in a straight line till reaches a point that is 20 m away from the source. Here, we consider two scenarios. In Scenario 1, we have two RISs located at the circumference of the circle constructing an isosceles triangle where its base is between the two RISs. The two equal angles here are $\theta_{1}=\theta_{2}=36.87^{o}$. In Scenario 2, we have two more RISs are added on the circumference of the circle constructing another isosceles triangle where its base is again between the two RISs. The two equal angles here are $\theta_{1}=\theta_{2}=11.53^{o}$. Each RIS is assumed to have the same number of reflecting elements, that is $N_{1}=N_{2}=N_{3}=N_{4}=15$. It is obvious from this figure that increasing the distance between the source and destination deteriorates the system performance, as expected. In addition, it can be seen from this figure that Scenario 2 is outperforming Scenario 1 as the second scenario has a better diversity order. These two scenarios are repeated again, but now when the RISs have more reflecting elements, that is $N_{1}=N_{2}=N_{3}=N_{4}=20$. Again, Scenario 2 is outperforming Scenario 1, but now with better results compared to the first set of curves where less number of reflecting elements were used.



\begin{figure}[]
	\includegraphics[scale=0.5]{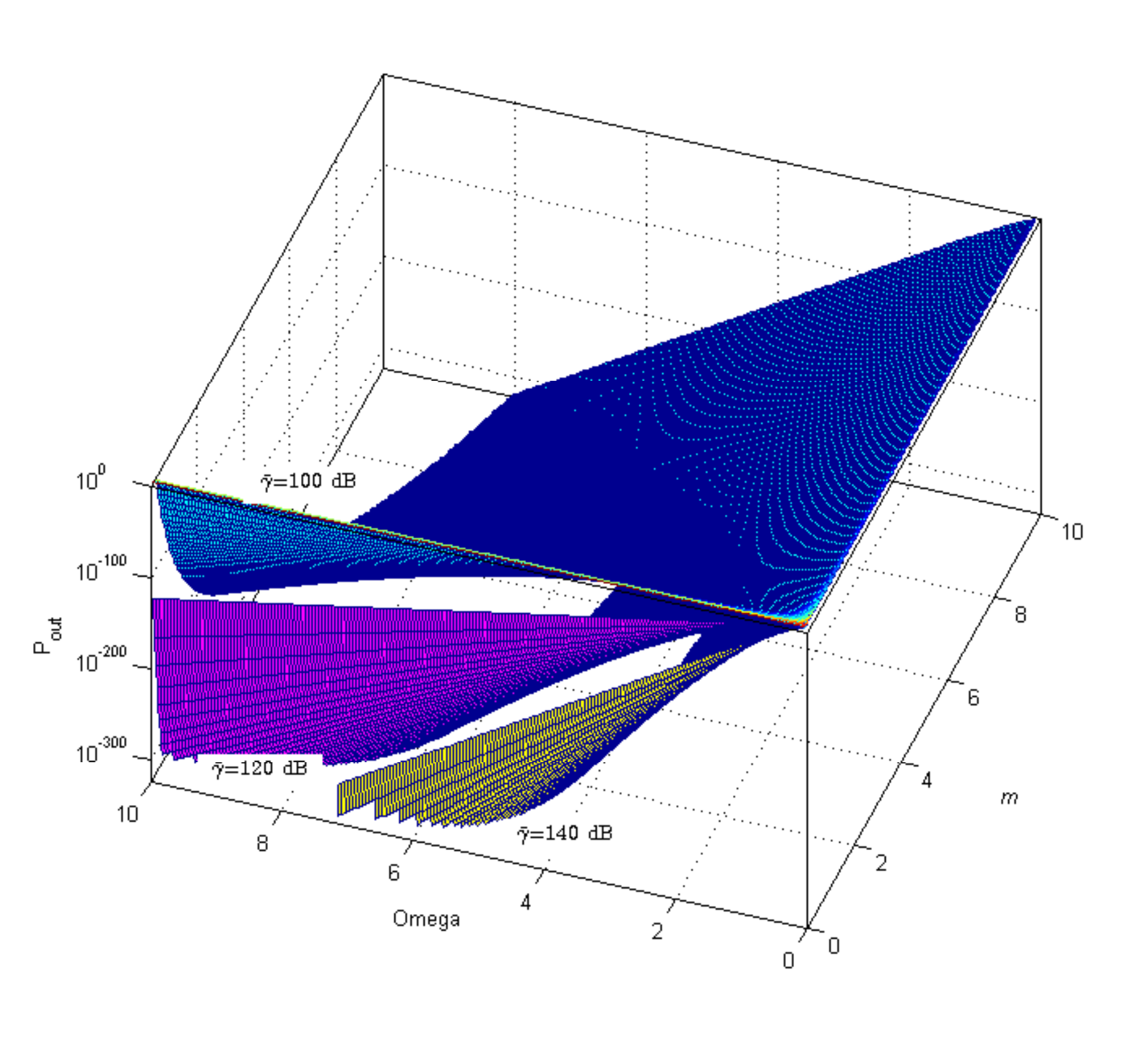}
	\centering
	\caption{The outage probability versus $m$ and $\Omega$ for different values of $\bar{\gamma}$. $K = 3$, $N_{1} = N_{2} = N_{3} = 5$, and $\left(d_{k,1}, d_{k,2}\right) = \left(5 \hspace{0.15cm}\text{m}, 5 \hspace{0.15cm}\text{m}\right)$.}	
	\label{threedimPout}
\end{figure}

The impact of channel scale parameter ($\Omega$) and channel shape parameter ($m$) on the system performance is studied in Fig. \ref{threedimPout}. In this figure, we assume the first hops and second hops of all RISs have the same channel scale parameter ($\Omega_{k,1}=\Omega_{k,2}=\Omega$). Similarly, we assume the first hops and second hops of all RISs have the same channel shape parameter ($m_{k,1}=m_{k,2}=m$). Clearly, we can see that increasing either $\Omega$ or $m$ enhances the system performance. It is important to mention here that $\Omega$ affects the system performance through the coding gain, whereas $m$ affects the performance through the diversity order.

\begin{figure}[]
	\includegraphics[scale=0.5]{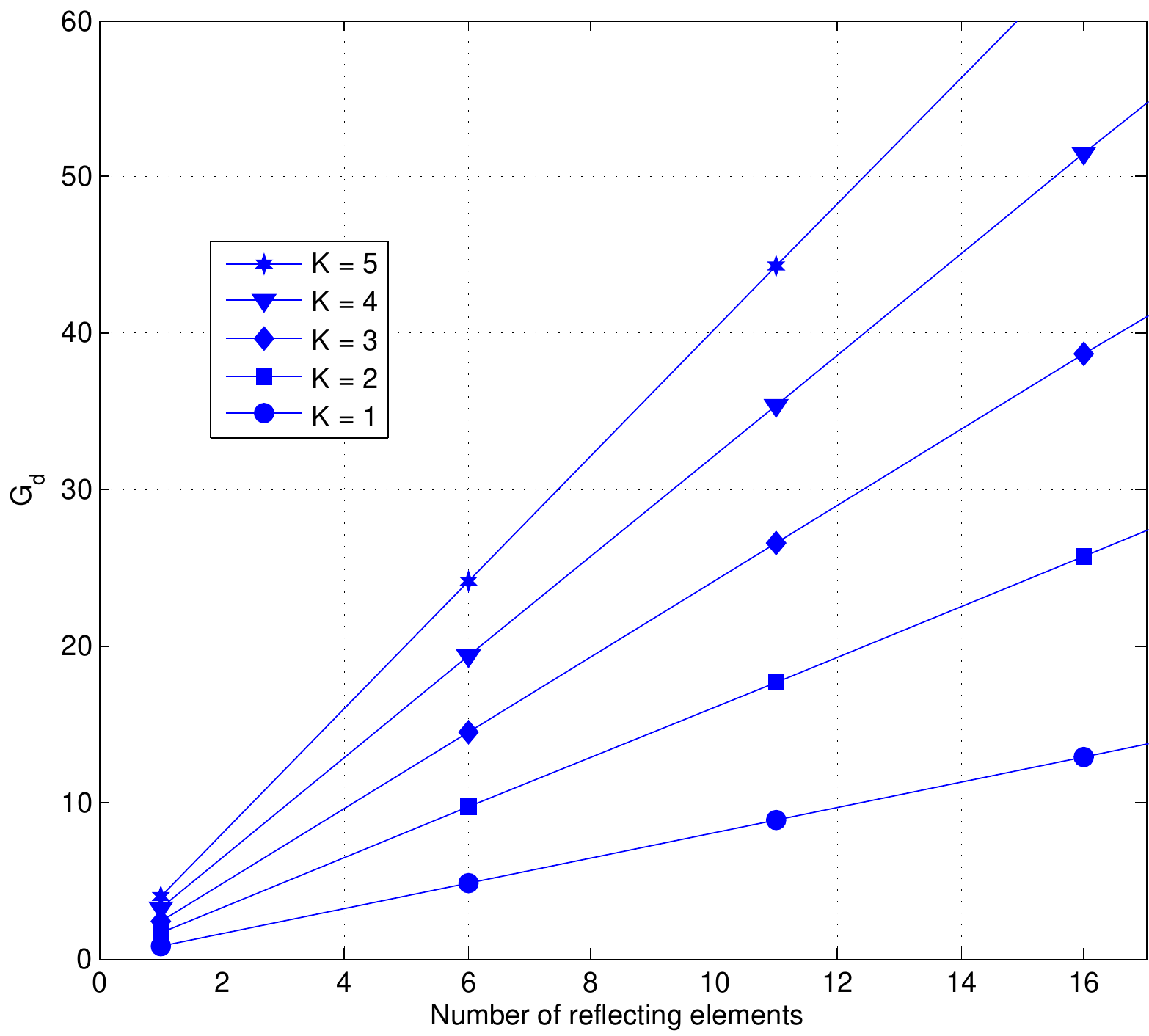}
	\centering
	\caption{The diversity gain versus numbers of reflecting RIS elements for different numbers of RISs. $\left(m_{k,1}, m_{k,2}\right) = \left(1, 1\right)$, $\left(\Omega_{k,1}, \Omega_{k,2}\right) = \left(1, 1\right)$, $\left(d_{k,1}, d_{k,2}\right) = \left(5 \hspace{0.15cm}\text{m}, 5 \hspace{0.15cm}\text{m}\right)$ and average SNR $\bar{\gamma} = 80$ dB.}	
	\label{diversitygain_N}
\end{figure}

A comparison between the effect of number of RISs ($K$) and number of reflecting elements ($N$) on the system diversity order is studied in Fig. \ref{diversitygain_N}. If we take the case where $N=5$ and $K=1$, we can see that the diversity order of the system here is $G_{d}\approx 5$. In addition, for the case where $N=1$ and $K=5$, we can see that the diversity order is also $G_{d}\approx 5$. This means that both system parameters $K$ and $N$ have almost the same effect on the system diversity gain.

The ASEP for two modulation schemes is shown in Fig. \ref{ASEP}: BPSK ($p = q = 1$) and QPSK ($p = 1, q = 0.5$). We can
see that the QPSK is outperforming the BPSK, as expected. This impact comes in terms coding gain. In addition, the impact of
increasing the channel scale parameter $m$ and the number of reflecting elements $N$ is clear here through increasing the system diversity order.

\begin{figure}[]
	\includegraphics[scale=0.5]{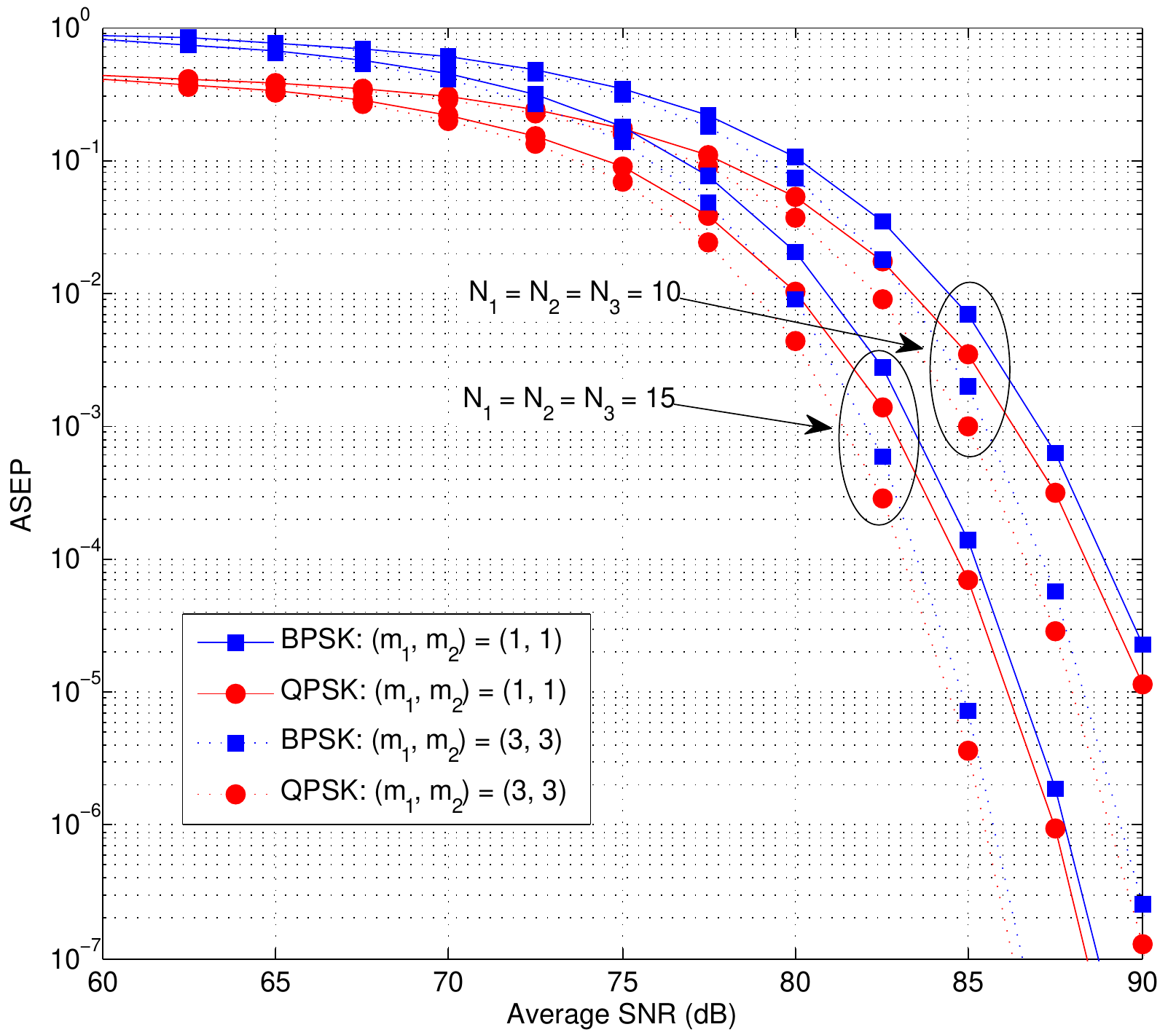}
	\centering
	\caption{The ASEP versus average SNR for BPSK and QPSK with different numbers of reflecting RIS elements and different values of Nakagami-\emph{m}. $\left(\Omega_{k,1}, \Omega_{k,2}\right) = \left(1, 1\right)$ and $\left(d_{k,1}, d_{k,2}\right) = \left(5 \hspace{0.15cm}\text{m}, 5 \hspace{0.15cm}\text{m}\right)$.}	
	\label{ASEP}
\end{figure}

In Fig. \ref{ASEP_identical_non}, we study a very important aspect in multi-RIS-aided networks that is the serial RIS configuration where the distance between the source and destination is fixed and the RISs are distributed between them in a serial way, and the parallel RIS configuration where the RISs are located somewhere at the midway between the source and destination. We also investigate in this figure the impact of reflecting elements' distribution where a fixed number of reflecting elements is one time equally distributed among the multiple RISs and in other time is distributed in a random (unequal) manner. In this figure, the distance between the source and destination is kept as 10 m and the total number of reflecting elements is also kept as 24 with three RISs. In Scenario 1, we locate the three RISs in a parallel configuration ($d_{1,1}=d_{2,1}=d_{3,1}=5$ m, $d_{1,2}=d_{2,2}=d_{3,2}=5$ m) all located somewhere in the midway between the source and destination (circle-marked solid-line). In Scenario 2, we locate the three RISs in a serial configuration ($d_{1,1}=2.5$ m, $d_{1,2}=7.5$ m and $d_{2,1}=5$ m, $d_{2,2}=5$ m and $d_{3,1}=7.5$ m, $d_{3,2}=2.5$ m) somewhere on a straight line between the source and destination (square-marked solid-line). Obviously, Scenario 2 is outperforming Scenario 1 in terms of coding gain as the distance impact is embedded in the path-loss ($P_{L,k}$), which in turn affects the coding gain of the system. It is worthwhile to mention here that as the total number of reflecting elements ($N$) is the same in both scenarios, the diversity order of the system stayed the same in these two figures. Another important aspect studied in this figure is the distribution of the total number of reflecting elements among the three RISs in both the parallel and serial configurations (circle and square-marked dotted-lines). It can be seen that the scenarios where the total number of reflecting elements is randomly distributed among the multiple RISs gives better performance than the scenario where the reflecting elements are equally distributed among the RISs. This is because the random distribution of reflecting elements gives better coding gain than the equal distribution scenario.

\begin{figure}[]
	\includegraphics[scale=0.5]{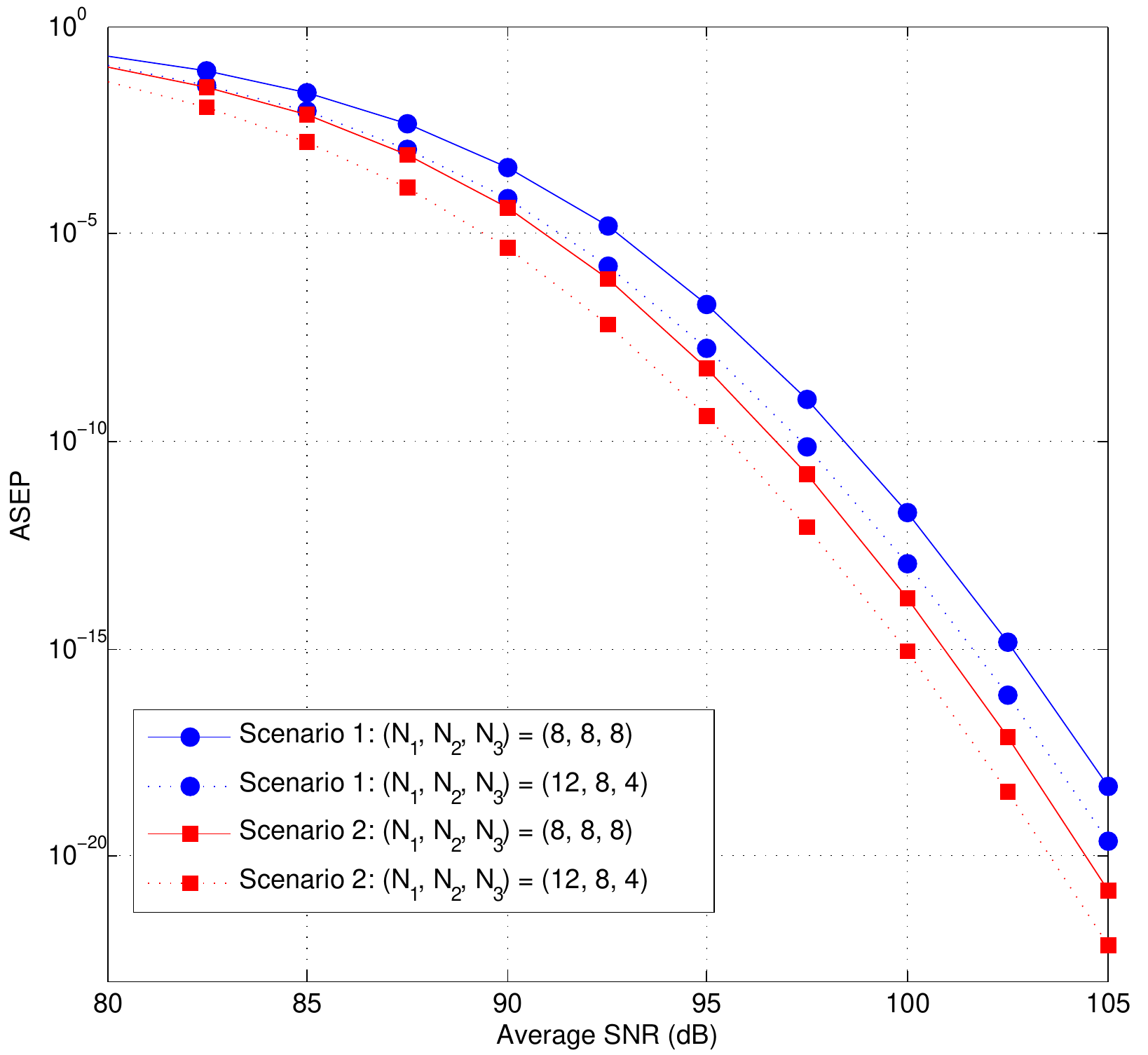}
	\centering
	\caption{The ASEP of BPSK versus average SNR for two scenarios. $K = 3$, $\left(m_{k,1}, m_{k,2}\right) = \left(1, 1\right)$ and $\left(\Omega_{k,1}, \Omega_{k,2}\right) = \left(1, 1\right)$.}	
	\label{ASEP_identical_non}
\end{figure}

\begin{figure}[]
	\includegraphics[scale=0.5]{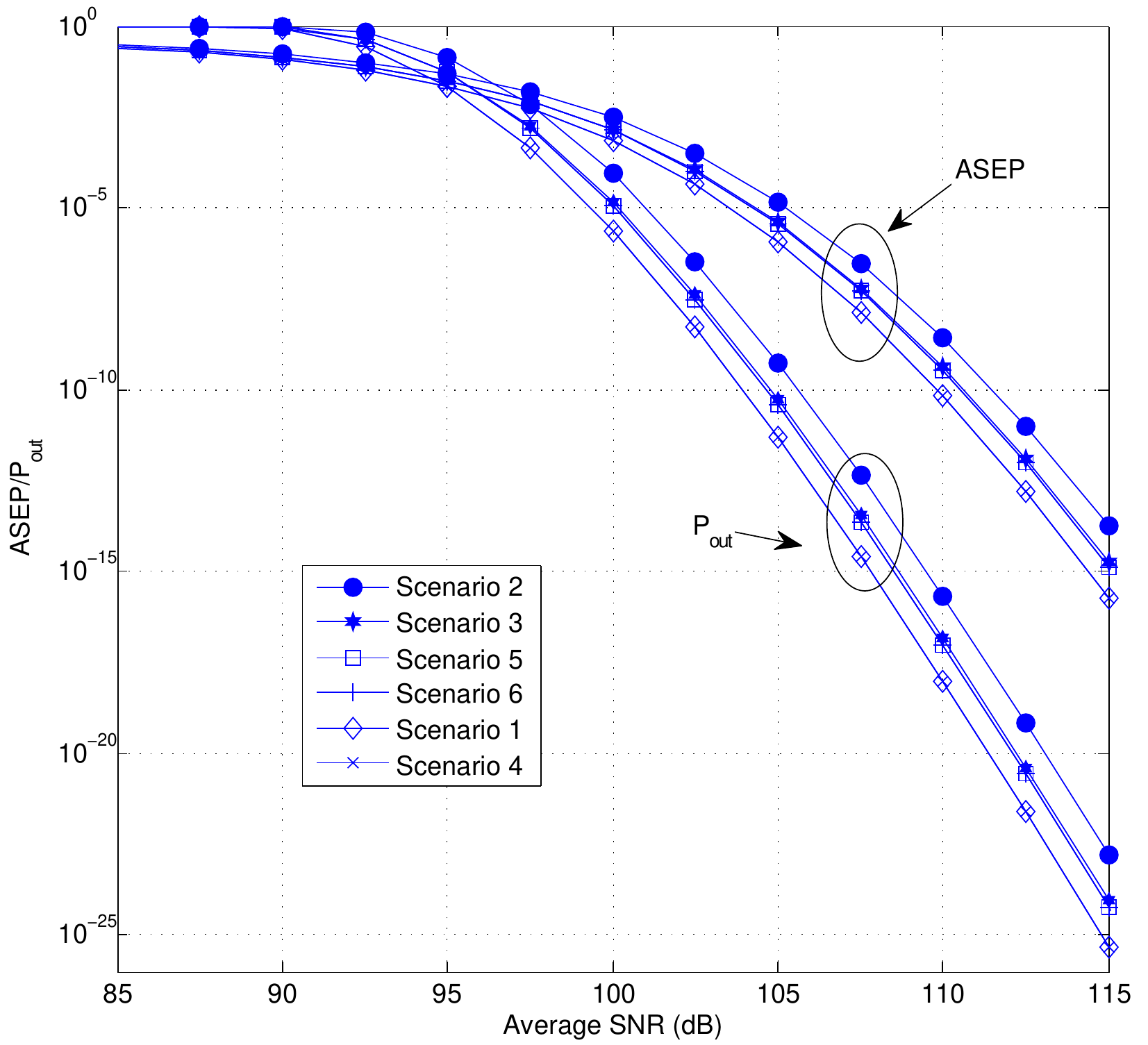}
	\centering
	\caption{The ASEP of BPSK and outage probability versus average SNR for various scenarios of RIS location. $K = 3$, $\left(m_{k,1}, m_{k,2}\right) = \left(1, 1\right)$, $\left(\Omega_{k,1}, \Omega_{k,2}\right) = \left(1, 1\right)$ and $\left(d_{k,1}, d_{k,2}\right) = \left(5 \hspace{0.15cm}\text{m}, 5 \hspace{0.15cm}\text{m}\right)$}	
	\label{ASEPPout_room}
\end{figure}

A very important aspect to study in multi-RIS-aided networks is the location of RISs with respect to source and destination locations. To address this very important point, we consider here a hall of $10\times 20$ m$^{2}$ area. We also consider that the height of the source, RISs, and destination from the ground is the same, say 1.5 m. Considering the hall plotted in a vertical way, that is the width on the $x$-axis and the length on the $y$-axis, we assume that the zero coordinates ($0,0$) are represented by the bottom left corner of the room (reference point). We also assume that the source is located at coordinates ($5,0$) and the destination at ($5,20$). This scenario could represent the practical condition where the source and destination are far away from each other inside the hall. Now, we study the impact of locating two RISs in different locations on side walls on the system performance in Fig. \ref{ASEPPout_room}. We consider for that purpose 7 different scenarios as follows: Scenario 1: RIS 1 at ($10,0$) and RIS 2 at ($0,0$), Scenario 2: RIS 1 at ($10,10$) and RIS 2 at ($0,10$), Scenario 3: RIS 1 at ($10,0$) and RIS 2 at ($10,10$) or RIS 1 at ($10,0$) and RIS 2 at ($0,10$), Scenario 4: RIS 1 at ($10,0$) and RIS 2 at ($10,20$) or RIS 1 at ($10,0$) and RIS 2 at ($0,20$), Scenario 5: RIS 1 at ($10,5$) and RIS 2 at ($0,15$) or RIS 1 at ($10,5$) and RIS 2 at ($10,15$), and Scenario 6: RIS 1 at ($10,5$) and RIS 2 at ($0,5$) or RIS 1 at ($10,15$) and RIS 2 at ($0,15$). According to the distances between the RISs to source and destination, we can divide these six scenarios into four groups from worst to best: Scenario 2, Scenario 3, Scenarios 5 and 6, and lastly Scenarios 1 and 4. Clearly, we can see that locating the two RISs in Scenario 2 at the horizontal line of symmetry on side walls between the source and destination results in the worst performance. Then, Scenario 3 comes where only one of the two RISs is located at the midway and the other RIS at the corner near either the source or destination. The third set of curves is for both Scenario 5 and 6 where the the two RISs are located either 5 m away from the source or destination on the vertical axis (on side walls). Finally, the best performance is achieved in Scenarios 1 and 4 where the two RISs are located 10 m away from either the source or destination, but on the same horizontal axis. Clearly, the performance in these six scenarios is determined by the distances between the RISs to source and destination, which in turn affects the path-loss, and hence, the coding gain of the system. In summary, for the situation where the source and destination are located at the far sides of the hall, that is the source at ($5,0$) and the destination at ($5,20$), it is recommended to have the two RISs located at the hall corners (any two corners of the four).

\begin{figure}[]
	\includegraphics[scale=0.28]{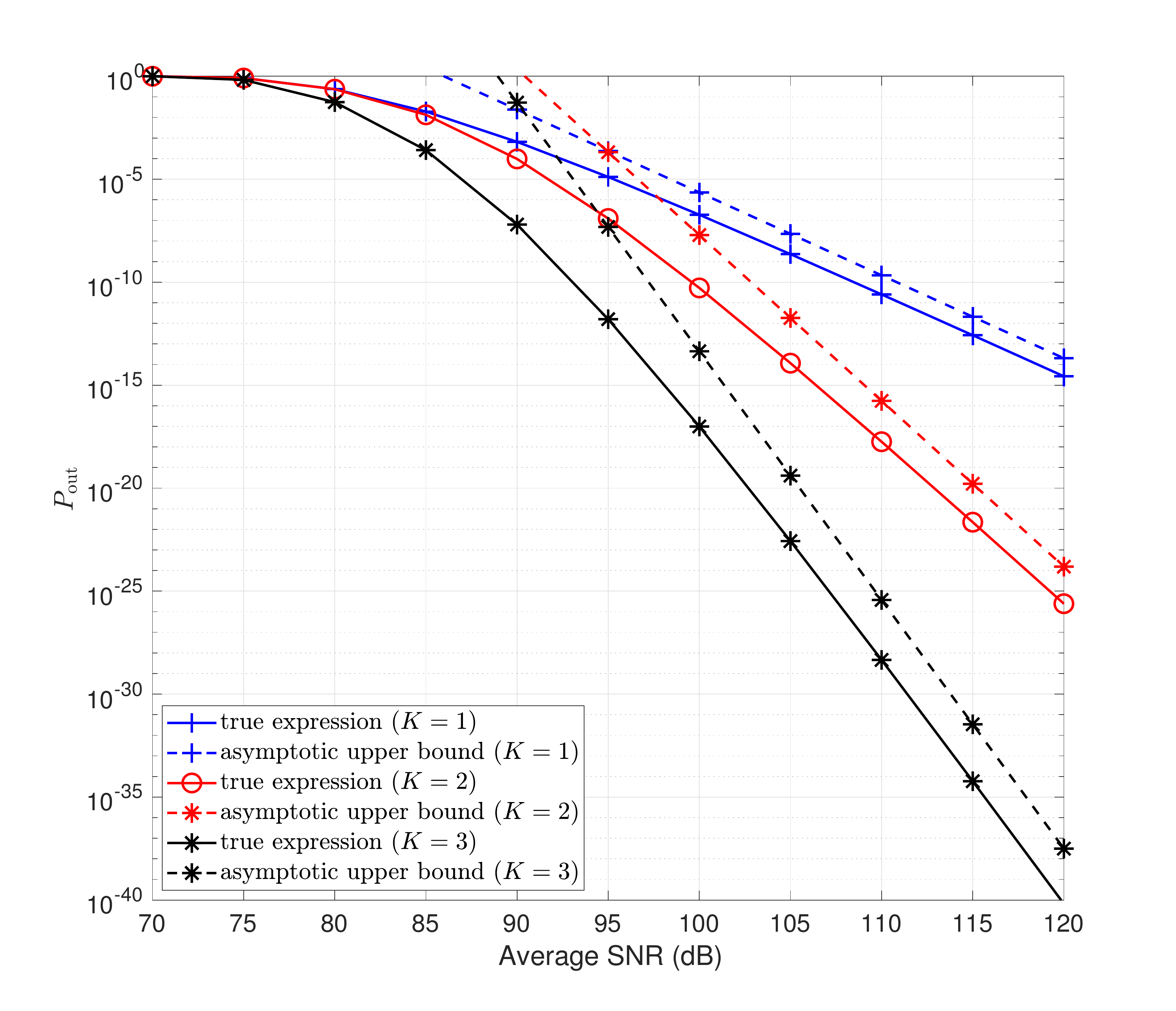}
	\centering
	\caption{The true expression and asymptotic upper bound of outage probability versus average SNR for $K = 3$, $\left(m_{k,1}, m_{k,2}\right) = \left(1, 1\right)$, $\left(\Omega_{k,1}, \Omega_{k,2}\right) = \left(1, 1\right)$, $\left(d_{1,1}, d_{1,2}\right) = \left(1 \hspace{0.15cm}\text{m}, 9 \hspace{0.15cm}\text{m}\right)$, $\left(d_{2,1}, d_{2,2}\right) = \left(5 \hspace{0.15cm}\text{m}, 5 \hspace{0.15cm}\text{m}\right)$, and $\left(d_{3,1}, d_{3,2}\right) = \left(9 \hspace{0.15cm}\text{m}, 1 \hspace{0.15cm}\text{m}\right)$,}	
	\label{Fig0}
\end{figure}

{\color{black}Finally, we provide some simulations for the optimization analysis, which will provide answer to the questions posed at the start of Section \ref{Opt}. In order to simulate results of this section, we set $\left(m_{k,1}, m_{k,2}\right) = \left(1, 1\right)$ and $\left(\Omega_{k,1}, \Omega_{k,2}\right) = \left(1, 1\right)$.  Since, these results are based on the asymptotic upper bound on the OP, which is provided by the expression \eqref{eq:AsympOut_UB}, we first plot the true expression and asymptotic upper bound of OP versus average SNR  in Fig. \ref{Fig0} for $K = 3$,  $\left(d_{1,1}, d_{1,2}\right) = \left(1 \hspace{0.15cm}\text{m}, 9 \hspace{0.15cm}\text{m}\right)$, $\left(d_{2,1}, d_{2,2}\right) = \left(5 \hspace{0.15cm}\text{m}, 5 \hspace{0.15cm}\text{m}\right)$, and $\left(d_{3,1}, d_{3,2}\right) = \left(9 \hspace{0.15cm}\text{m}, 1 \hspace{0.15cm}\text{m}\right)$, i.e., the three RISs are placed at different places, one close to the S, one mid-way, and one near to the D. It can be observed from Fig. \ref{Fig0} that the UB-OP \eqref{eq:AsympOut_UB} provides an upper bound on the true expression asymptotically.

\begin{figure}[]
	\centering
	\includegraphics[width=0.5 \textwidth]{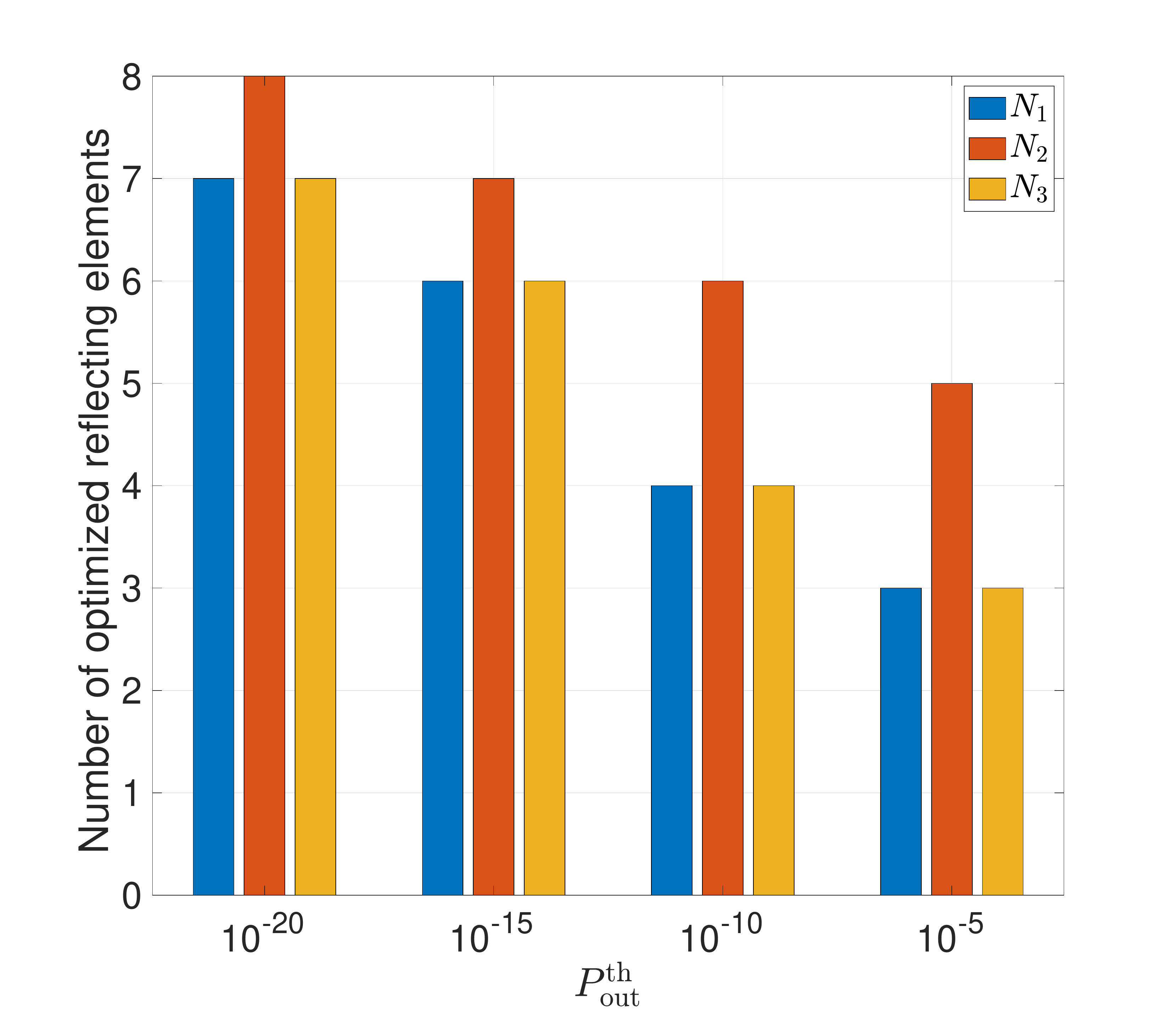}
	\caption{Optimized number of reflecting elements versus different values of threshold outage probability by solving problem \eqref{p1} using Alg. \ref{alg1} for $K = 3$, $\bar{\gamma} = 100$ dB, $\left(m_{k,1}, m_{k,2}\right) = \left(1, 1\right)$, $\left(\Omega_{k,1}, \Omega_{k,2}\right) = \left(1, 1\right)$, $\left(d_{1,1}, d_{1,2}\right) = \left(1 \hspace{0.15cm}\text{m}, 9 \hspace{0.15cm}\text{m}\right)$, $\left(d_{2,1}, d_{2,2}\right) = \left(5 \hspace{0.15cm}\text{m}, 5 \hspace{0.15cm}\text{m}\right)$, and $\left(d_{3,1}, d_{3,2}\right) = \left(9 \hspace{0.15cm}\text{m}, 1 \hspace{0.15cm}\text{m}\right)$.}
	\label{Fig1}
\end{figure}

Fig. \ref{Fig1} plots the optimized number of reflecting elements versus different values of threshold outage probability $P_\text{out}^\text{th}$ by solving problem \eqref{p1} using Alg. \ref{alg1} for $K = 3$, $\left(d_{1,1}, d_{1,2}\right) = \left(1 \hspace{0.15cm}\text{m}, 9 \hspace{0.15cm}\text{m}\right)$, $\left(d_{2,1}, d_{2,2}\right) = \left(5 \hspace{0.15cm}\text{m}, 5 \hspace{0.15cm}\text{m}\right)$, and $\left(d_{3,1}, d_{3,2}\right) = \left(9 \hspace{0.15cm}\text{m}, 1 \hspace{0.15cm}\text{m}\right)$. As expected, the optimal number of reflection elements decreases with the increase in $P_\text{out}^\text{th}$. Note that the optimal value of $N_2$ is higher than that of $N_1$ and $N_2$. This gives us a hint that more reflecting elements are required by the mid-way placed RIS compared to those, which are placed close to the S or D. This is because the signal reflected by the mid-way placed RIS  experiences larger effective path-loss. This finding is strongly supported by Fig. \ref{Fig2}, which plots the optimized number of reflecting elements by solving problem \eqref{p4} using Alg. \ref{alg1}. Since the objective is to minimize the total number of reflecting elements, the optimal solution is to use a single RIS, which is placed close to the D. Note that due to non-convexity of the original problem \eqref{p4}, the solution depends on the initialization. With a different initialization, we may observe that the optimal solution by Alg. \ref{alg1} is to select a single RIS that is placed close to the S. However, the mid-way placed RIS is never selected because it requires the deployment of larger number of RISs. 

We also solve the problem \eqref{p7} to determine the optimal relative placement of RISs, which is iteratively solved by solving the linearized problem \eqref{p8}. As expected from the analysis provided below \eqref{p7} and from the above explanation, the optimal placement of RISs was at the extreme points, either near to the S or D, depending on the initialized value of the distance vector.
}

\begin{figure}[]
	\centering
	\includegraphics[width=0.5 \textwidth]{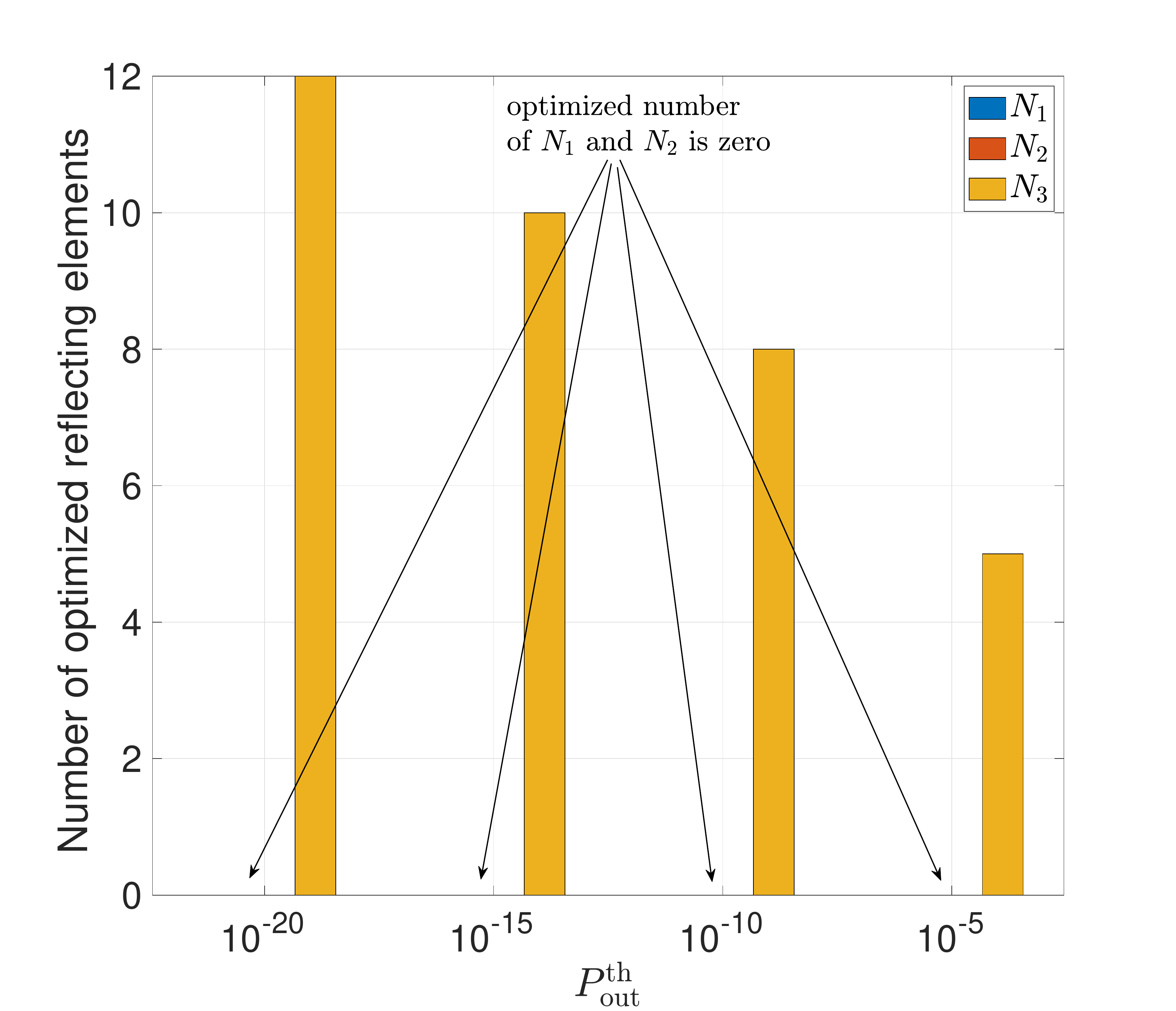}
	\caption{Optimized number of reflecting elements versus different values of threshold outage probability by solving problem \eqref{p4} using Alg. \ref{alg1} for $K = 3$, $\bar{\gamma} = 100$ dB, $\left(m_{k,1}, m_{k,2}\right) = \left(1, 1\right)$, $\left(\Omega_{k,1}, \Omega_{k,2}\right) = \left(1, 1\right)$, $\left(d_{1,1}, d_{1,2}\right) = \left(1 \hspace{0.15cm}\text{m}, 9 \hspace{0.15cm}\text{m}\right)$, $\left(d_{2,1}, d_{2,2}\right) = \left(5 \hspace{0.15cm}\text{m}, 5 \hspace{0.15cm}\text{m}\right)$, and $\left(d_{3,1}, d_{3,2}\right) = \left(9 \hspace{0.15cm}\text{m}, 1 \hspace{0.15cm}\text{m}\right)$.}
	\label{Fig2}
	\vspace{-0.2in}
\end{figure}
\section{Conclusion}\label{C}
This paper studied the performance of multiple RISs-aided networks over Nakagami-$m$ fading channels. In addition, an optimization problem where the optimum number of reflecting elements that guarantee a predetermined outage performance and the optimal relative placement of RISs was formulated and solved. For that purpose, accurate closed-form approximations were first derived for the channel distributions and then used in deriving closed-form approximations for the OP and ASEP assuming i.ni.d channels. Furthermore, to get more insights into the system performance, an asymptotic expression was derived for the OP at the SNR regime and provides closed-form expressions for the system diversity order and coding gain. Results showed that the considered RIS scenario can provide a diversity order of $\left(\frac{a}{2}\right)K$, where $a$ is a function of the Nakagami-$m$ fading parameter $m$, and the number of meta-surface elements $N$, and $K$ is the number of reflecting elements. Furthermore, results showed that distributing the same number of reflecting elements on multiple RISs decreases the system performance compared to assigning them to a single RIS. In addition, random distribution of reflecting elements on several RISs was shown to give better performance than equal distribution. Finally, some new scenarios on RIS location with respect to source and destination and its impact on the system performance were discussed in this work.

	\end{document}